\documentclass[11pt, oneside]{article} 
\usepackage[margin=2cm]{geometry}                		
\usepackage[parfill]{parskip}    		
\usepackage{graphicx}			
\usepackage[utf8]{inputenc}
\usepackage[english]{babel} % English language/hyphenation
\usepackage{amsmath,amssymb,amsthm,bm} % Math packages
\usepackage[margin=2cm]{geometry}
\usepackage[color=yellow]{todonotes}
\usepackage{mathrsfs}
\usepackage{url}	
\usepackage{esint}
\usepackage[colorlinks]{hyperref}
\usepackage{enumerate}
\usepackage{outlines}[enumerate]
\usepackage{framed,comment,enumerate}
\usepackage{upgreek}
\usepackage{pgfplots}
\usepackage{float}
\usepackage{tikz,tikz-cd}
\usepackage{rotating}
\usepackage{graphicx}
\usepackage{caption}
\usepackage{multicol}
\usepackage{cancel}
\usepackage{subcaption}
\usepackage[title]{appendix}
\usepackage{tikz-cd}
\usepackage{pgfgantt}
\usepackage{rotating}
\extrafloats{100}
\numberwithin{equation}{section}

\newtheorem{proposition}{Proposition}[section]
\newtheorem{remark}{Remark}[section]

\theoremstyle{definition}

\DeclareFontFamily{U}{MnSymbolC}{}
\DeclareSymbolFont{MnSyC}{U}{MnSymbolC}{m}{n}
\DeclareFontShape{U}{MnSymbolC}{m}{n}{
    <-6>  MnSymbolC5
   <6-7>  MnSymbolC6
   <7-8>  MnSymbolC7
   <8-9>  MnSymbolC8
   <9-10> MnSymbolC9
  <10-12> MnSymbolC10
  <12->   MnSymbolC12}{}
\DeclareMathSymbol{\intprod}{\mathbin}{MnSyC}{'270}

\newcommand{\wt}[1]{\widetilde{#1}}
\newcommand{\wh}[1]{\widehat{#1}}

\newcommand{\mc}[1]{\mathcal{#1}}

\newcommand{\mcal}[1]{\mc{#1}}
\newcommand{\scp}[2]{{\left\langle {#1}\, , \, {#2}\right\rangle}}

\def\p{{\partial}}

\def\bm{{\mathbf{m}}}

\def\bu{{\mathbf{u}}}
\def\bw{{\mathbf{w}}}

\def\bv{{\mathbf{v}}}
\def\bx{{\mathbf{x}}}

\def\p{\partial}

\pgfplotsset{compat=1.16}

\def\bm{\mathbf{m}}

\def\bu{\mathbf{u}}
\def\bv{\mathbf{v}}
\def\bx{\mathbf{x}}

\title{Surface Wave Solutions in 1D and 2D for the Broer-Kaup-Boussinesq-Kupershmidt (BKBK) System}
\author{{\Large Darryl D. Holm$^1$\footnote{Corresponding author. Email: d.holm@imperial.ac.uk, ORCID ID: 0000-0001-6362-9912}, \, Ruiao Hu$^1$\footnote{Email: ruiao.hu15@imperial.ac.uk, ORCID ID: 0000-0002-4843-1737}\vspace{0.4cm}~~ and Hanchun Wang$^2$\footnote{Email: hanchun.wang21@imperial.ac.uk, ORCID ID: 0000-0002-9669-1322}} \\ 
1. Department of Mathematics, Imperial College London, \\ London SW7 2AZ, UK\\
2. DAMTP, University of Cambridge, \\}
\date{}

\begin{document}
\maketitle

\centerline{\large In memoriam D.J. Kaup (1939-2022) and B.A. Kupershmidt (1946-2010)}

\begin{abstract}
%The celebrated Boussinesq-Kaup-Broer-Kupershmidt system for surface wave dynamics is abbreviated here as BKBK. In his 1985 paper, Boris Kupershmidt called BKBK ``the richest integrable system known to date'' and proceeded to derive three compatible Poisson operators for its Hamiltonian formulation. 

The BKBK system is a singular perturbation of the classical shallow water equations which modifies their transport velocity to depend on wave elevation slope. This modification introduces backward diffusion terms proportional to a real parameter $\kappa$. These terms also make BKBK completely integrable as a Hamiltonian system. Remarkably, when $\kappa=i/2$ the BKBK system may be transformed into the focusing nonlinear Schr\"odinger (NLS). Thus, the BKBK system with its real parameter $\kappa$ is complementary to the traditional modulational approach for water waves.
We investigate the Lie algebraic and variational properties of the BKBK system in this paper and we study its solution behaviour in certain computational simulations of regularised versions of the 1D and 2D BKBK systems. 

%We also develop some additional theoretical methods for interpreting regularised BKBK simulations, including deriving spectral conditions for linear Lyapunov stability of equilibrium solutions.

\end{abstract}

\tableofcontents

%%%%%%%%%%%%%%%%%%%%%%%%%%%%%%%%%%%%%%
\section{Introduction}\label{sec-1}
%%%%%%%%%%%%%%%%%%%%%%%%%%%%%%%%%%%%%%%%%%%%%%%%%%%%%%%%
%{\color{blue}\large\fbox{\bf This section does what?} }
%%%%%%%%%%%%%%%%%%%%%%%%%%%%%%%%%%%%%%%%%%%%%%%%%%%%%%%%

\subsection{BKBK system in one dimension (1D)}\label{sec-BKBK-1D}
The classical shallow water equations in one dimension (1D) for fluid velocity $ u = u(x,t)$ and depth $\eta = \eta(x,t)$ on the real line $x\in \mathbb{R}$ are given by
\begin{align}\begin{split}
u_t+u u_x + \frac{1}{Fr^2}\eta_x &=0\,, 
\\
\eta_t+(u \eta)_x &=0
\,,\label{SW-system} 
\end{split}
\end{align} 
in which subscripts in $t$ and $x$ denote partial derivatives, $Fr^2=u_0^2/g\eta_0$ is the dimension-free Froude number, in which $g$ is the gravitational constant, $u_0$ is the mean velocity and $\eta_0$ is the mean depth for Boussinesq long waves in shallow water. In what follows, we will take $Fr = 1$. The dispersion relation $\omega(k)$  for linearised wave solutions of the 1D shallow water system in \eqref{SW-system} proportional to $\exp(i(kx-\omega t))$ with frequency $\omega$ and wavenumber $k$ is then given by
\begin{equation}
\omega^2 = k^2 
.
\label{disp-SW}
\end{equation}
Since the phase velocity $\omega/k=\pm 1$ of linearised waves travelling on the real axis is independent of wave number $(k)$, the classical 1D shallow water system \eqref{SW-system} is said to be \emph{dispersionless}.

Historically, the shallow water system \eqref{SW-system} has had a number of interesting dispersive modifications in both 1D and 2D, going back at least to  Boussinesq \cite{boussinesq1872theorie}. These modifications are reviewed, e.g., in Broer \cite{broer1975approximate} and Kaup \cite{kaup1975higher}, as well as in more recent compendia of references and discussions, such as \cite{cheviakov2024analytical,klein2021nonlinear,klein2025kaup}.  

In this work, we study the following integrable dispersive singular perturbation of the nonlinear long-wave systems in the Boussinesq class, 
\begin{align}\begin{split}
u_t &=-\left(u^2 / 2+\eta+\beta u_x\right)_x 
\,,\\
\eta_t &=-\left(u \eta +\alpha u_{x x}-\beta \eta_x\right)_x
\,,\label{BK-System-BAK2} 
\end{split}
\end{align} 
with arbitrary constants $\alpha$ and $\beta$. For values $\alpha=1 / 3$, $\beta=0$, the system \eqref{BK-System-BAK2}  was derived by Broer  \cite{broer1975approximate} who called it ``The oldest, simplest and most widely known set of equations" for dispersive nonlinear long wave propagation.  
Kupershmidt \cite{kupershmidt1985mathematics} reduced the number of free parameters in the dispersive nonlinear long-wave system \eqref{BK-System-BAK2} by using the invertible change of variables,
\begin{align*}
    u\to u\,,\ \quad \eta\to \eta + \gamma u_x\,, \quad \text{where} \quad \gamma := - \beta \pm \sqrt{\alpha+\beta^2}
    \,,
\end{align*}
to obtain what we call here the Broer-Kaup-Boussinesq-Kupershmidt (BKBK) integrable system,
\begin{align}\begin{split}
{u}_t &= - \left({u}^2 / 2+{\eta}+\kappa {u}_x\right)_x = - uu_x - \eta_x - \kappa u_{xx}
\,,\\
{\eta}_t &= - ({\eta}v)_x  = - (\eta u )_x + \kappa {\eta}_{xx}
\quad\hbox{with}\quad 
v:= u -\kappa (\ln {\eta})_x
\,,\\
 \kappa &= \pm \sqrt{\alpha+\beta^2}
\,.\label{BK-System-BAK3} 
\end{split}
\end{align} 
The real parameter $\kappa$ appearing in \eqref{BK-System-BAK3} may be regarded as an arbitrary constant, say $\kappa = \pm 1 / 2$, after appropriately rescaling $x$ and $t$. 
Kupershmidt \cite{kupershmidt1985mathematics} also found that the system \eqref{BK-System-BAK3} for $\kappa=-1/2$ possesses three inequivalent compatible Hamiltonian structures and that this tri-Hamiltonian structure implies complete integrability of the system \eqref{BK-System-BAK3} for the case $\kappa=-1/2$. 

However, the backwards diffusion in the motion equation of system \eqref{BK-System-BAK3} produces an ill-posed dispersion relation $\omega^2(k^2)$ for the frequency $\omega$ as a function of wavenumber $k$ of the linearised BKBK solutions, 
\begin{equation}
\omega^2(k^2) = k^2(1 - \kappa^2 k^2) 
\,.
\label{disp-noalpha}
\end{equation}
Thus, the BKBK system \eqref{BK-System-BAK3} is linearly ill-posed for higher wave numbers, $ k^2>1 / \kappa^2 $, independently of the sign of the real parameter $\kappa$. Hence, regardless of its complete integrability as a Hamiltonian system for $\kappa = -1/2$, the linear ill-posedness of the 1D BKBK system is quite challenging to simulate numerically, since the solutions are linearly unstable for \emph{either} sign of $\kappa$ at higher wave numbers $ k^2>1 / \kappa^2 $, as discussed in \cite{klein2021nonlinear,klein2025kaup}. 

\textbf{Plan of the paper.} The remainder of the present investigation will proceed as follows.
\begin{itemize}
\item 
Section \ref{sec-2} discusses the geometric properties of the BKBK system in 1D. We present the various existing Hamiltonian structures of the 1D BKBK system \eqref{BK-System-BAK3} and its Euler--Poincar\'e derivation using symmetry-reduced variational principles \cite{HoMaRa1998a}. We investigate the Lyapunov stability of 1D BKBK system using the energy--Casimir approach of \cite{holm1985nonlinear} and connect the 1D BKBK system to the focusing nonlinear Schr\"odinger equation by changing the BKBK dispersion parameter $\kappa$ from real to imaginary.
\item Section \ref{sec-3} derives Euler--Poincar\'e formulation and the Hamiltonian formulation of the BKBK system in two spatial dimensions, as well as its Lie--Poisson bracket, Hamiltonian structures and conservation of potential vorticity (PV). 

\item Section \ref{sec-4} derives equilibrium conditions for the BKBK system in two dimensions following the stability results of \cite{holm1985nonlinear}. This is done by considering critical points of the sum $h_\Phi(\mathbf{u},\eta) = h_\Phi(\mathbf{u},\eta) + C_\Phi$ of the Hamiltonian found in section \ref{sec-3} plus its Casimir constants of motion, $C_\Phi$ in \eqref{DSW-Casimir}. By demanding positivity the second variation of $H_C$ we then derive spectral conditions for linear Lyapunov stability of the corresponding energy-Casimir class of equilibrium solutions of the 2D BKBK system.

\item Section \ref{sec-5} provides computational simulations of the solution behaviour for the BKBK system in both 1D and 2D. Different computational methods are employed for numerically regularising the BKBK systems in 1D and 2D. In 1D, we use a $4^{th}$-order dissipation reminiscent of the Kuramoto--Sivashinsky equation for stability. In 2D, we use a Hamiltonian regularisation where nonlinear dispersions are imposed by introducing an modified Hamiltonian that introduce energy costs for large wave height gradients.

\item Section \ref{sec-6} provides a summary conclusion and outlook for future research.
\end{itemize}
 
%%%%%%%%%%%%%%%%%%%%%%%%%%%%%%%%%%%%%%
\section{BKBK system in 1D} \label{sec-2}
%%%%%%%%%%%%%%%%%%%%%%%%%%%%%%%%%%%%%%%%%%%%%%%%%%%%%%%%
% This section reframes the 1D BKBK system \cite{kupershmidt1985mathematics} as a Lie--Poisson Hamiltonian system and then explores its relation with modulation instability, by using the inverse Madelung transform to identify the 1D BKBK system with the focusing nonlinear Schr\"odinger equation when the real parameter $\kappa$ appearing in \eqref{BK-System-BAK3} takes the imaginary value $\kappa=i/2$.
%%%%%%%%%%%%%%%%%%%%%%%%%%%%%%%%%%%%%%%%%%%%%%%%%%%%%%%%

The 1D dispersive BKBK system in \eqref{BK-System-BAK3} possesses three inequivalent compatible Hamiltonian formulations \cite{kaup1975higher,kupershmidt1985mathematics,clamondwater}. One of these is the well-known constant-coefficient Hamiltonian structure for the fluid velocity $u$ and depth $\eta$ that is defined through the Poisson bracket $\{\,\cdot\,,\, \cdot\,\}_c$. For arbitrary functions $F, G$ of the variables $(u,\eta)$, $\{\,\cdot\,,\, \cdot\,\}_c$ is defined by 
\begin{align}
    \{F,G\}_c(u,\eta) = \int_{\mathbb{R}}\left[ \frac{\delta G}{\delta \eta} \p_x\frac{\delta F}{\delta u} - \frac{\delta F}{\delta \eta}\p_x \frac{\delta G}{\delta u} \right]\,dx  \,.\label{CPB-def} 
\end{align}
Being a bilinear, antisymmetric, constant-coefficient differential operator, the Poisson bracket $\{\,\cdot\,,\, \cdot\,\}_c$ 
also satisfies the Jacobi identity,
\begin{equation}
\{A\,,\,\{B\,,\,D\}_c\}_c + \{B\,,\,\{D\,,\,A\}_c\}_c + \{D\,,\,\{A\,,\,B\}_c\}_c   = 0
\,,\label{Jacobi Id} 
\end{equation}
for arbitrary functionals $A, B, D$ of the variables $u$ and $\eta$. The time evolution of an arbitrary functional $F$ of $(u,\eta)$ is generated by the Hamiltonian functional $H$ on the same variables through the $\{\cdot\,,\,\cdot\}_c$ by
\begin{align}
    \frac{d}{dt}F(u,\eta) = \{F, H\}_c(u,\eta)\,,
\end{align}
such that the evolution of $(u,\eta)$ can be expressed in Poisson operator form as
\begin{equation}
\p_t 
\begin{bmatrix}
u \\ \eta
\end{bmatrix}
= -
\begin{bmatrix}
0 & \p_x \\
\p_x & 0 
\end{bmatrix}
\begin{bmatrix}
\delta {H}/\delta u 
\\  
\delta {H}/\delta \eta 
\end{bmatrix}
,\label{BKBK-Poisson1} 
\end{equation}
with 1D BKBK Hamiltonian $ H (u,\eta)$ and variations $\frac{\delta H}{\delta u}$ and $\frac{\delta H}{\delta \eta}$ given by
\begin{align}
\begin{split}
    H (u,\eta) &:= \int_{\mathbb{R}} \tfrac12 \eta u^2 - \kappa u \p_x \eta + \tfrac{1}{2}\eta^2\, dx \quad \text{such that}\\
\frac{\delta H}{\delta u} &= \eta v \quad\hbox{with}\quad v :=u - \kappa \p_x\ln \eta
\,,\quad 
\frac{\delta H}{\delta \eta} = \frac{1}{2}u^2 + \frac{\eta}{2} + \kappa \p_x u
\,.
\end{split}
\label{BKBK-Ham1}
\end{align} 
The operation of the matrix Poisson operator in \eqref{BKBK-Poisson1} on the variational derivatives of $H$ yields the 1D BKBK system \eqref{BK-System-BAK3}. From the Hamiltonian \eqref{BKBK-Ham1}, we see that the sign-indefinite terms multiplying $\kappa$ correspond to the terms that generate the singular dispersive perturbations of the classical shallow water equations appearing in the BKBK system \eqref{BK-System-BAK3}. 

\begin{remark}
Equilibrium solutions of the BKBK system \eqref{BK-System-BAK3} satisfy,
\begin{align}
\delta H_C = 0
\,,\quad \hbox{for} \quad H_C := {H} (u,\eta) +  C \int  u\,dx = 0
\,,\label{Eqm-conds1}
\end{align}
for an arbitrary constant speed $C$.
The second variation $\delta^2 {H}_C (u,\eta)$ is the Hamiltonian for the linearised flow in an infinitesimal neighbourhood of an equilibrium solution satisfying $\delta H_C = 0$ \cite{holm1985nonlinear}. Consequently, the Lyapunov stability of these BKBK equilibria may be tested by considering whether the following quadratic form $\delta^2 {H}_C (u,\eta)$ for the BKBK system is definite in sign when evaluated at an equilibrium solution $(u_e(x),\eta_e(x))$. 
\begin{align}
\delta^2 {H}_C (u,\eta)\Big|_{(u_e(x),\eta_e(x))} 
&= \int_{\mathbb{R}} 
\begin{bmatrix}
\delta u \\ \delta \eta
\end{bmatrix}^T
\begin{bmatrix}
\eta_e & u_e 
\\
u_e & 1  
\end{bmatrix}
\begin{bmatrix}
\delta u \\ \delta \eta 
\end{bmatrix}dx
+ 2 \kappa \int_{\mathbb{R}}  \delta u_x \delta \eta \,dx
\\&= 
\int_{\mathbb{R}} 
\begin{bmatrix}
\delta u \\ \delta \eta
\end{bmatrix}^T
\begin{bmatrix}
\eta_e & u_e -  \kappa\p_x
\\
u_e + \kappa\p_x  & 1  
\end{bmatrix}
\begin{bmatrix}
\delta u \\ \delta \eta 
\end{bmatrix}dx
\,.\label{2nd-var-conds}
\end{align}
For the case of 1D shallow water dynamics arising for $\kappa=0$, definiteness in sign of $\delta^2 H_C (\mathbf{u}_e,\eta_e) $ provides sufficient
conditions for linear Lyapunov stability of shallow water equilibrium solutions that are obtained from setting $\delta H_C (\mathbf{u}_e,\eta_e) =0$ with $\kappa =0$. This was shown for 2D shallow water dynamics in \cite{holm1983nonlinear,holm1985nonlinear}. However, for non-zero $\kappa$, the linear stability of perturbed 1D BKBK equilibrium solutions in the class $\delta H_C ({u}_e,\eta_e) =0$ requires the additional condition that the spectrum of the symmetric elliptic operator in \eqref{2nd-var-conds} must also be positive definite. 
However, the symbol of the elliptic operator in \eqref{2nd-var-conds} is of the form
\begin{align}
\sigma(k^2) = {\eta}_e(x) - \big( {u}_e^2(x) + \kappa^2\,k^2 \big) 
\,,
\end{align}
which cannot remain positive at arbitrarily high wavenumber $(k)$ for non-zero $\kappa$. 
Hence, the linear Lyapunov stability condition of positivity of the 2nd variation cannot be enforced at arbitrarily high wavenumber, $(k)$. Once high wave numbers were introduced in perturbing 1D BKBK equilibrium solutions such as its travelling wave solution, the equilibrium would no longer be linearly Lyapunov stable and the nonlinear terms would generate even higher wave numbers. 
\end{remark}
\paragraph{Connection to Nonlinear Schr\"odinger equation}\label{rem:NLS}
    The 1D BKBK system \eqref{BK-System-BAK3} can be expressed in terms of the modified transport velocity $v$. In this case, we have 
\begin{align}
    \begin{split}
        v_t = - \left(\frac{v^2}{2} - \frac{\kappa^2}{2}\left(\frac{\eta_x}{\eta}\right)^2 + \kappa^2\frac{\p_x^2\eta}{\eta}  
        + \eta \right)_x\,,\quad \eta_t = -\,(\eta v)_x\,.
    \end{split}\label{1d BKBK v eta}
\end{align}
System \eqref{1d BKBK v eta} can be expressed using the constant-coefficient Poisson bracket \eqref{BKBK-Poisson1} in combination with the following modified Hamiltonian $\mathfrak{h}$,
\begin{align}
    \mathfrak{h}(v,\eta) = \int_{\mathbb{R}} \frac{1}{2}\eta v^2 - \frac{\kappa^2}{2\eta}(\p_x\eta)^2  +\frac{1}{2}\eta^2 \,dx = \int_{\mathbb{R}} \frac{1}{2}\eta v^2 - \kappa^2(\p_x\sqrt{\eta})^2 + \frac{1}{2}\eta^2 \,dx\,.
\label{Ham-v}
\end{align}
One may notice that the term $(\p_x\sqrt{\eta})^2$ in the second expression for the Hamiltonian here is in the form of the Fisher--Rao metric of $\eta$. It has not escaped our attention that the Hamiltonian $\mathfrak{h}$ in \eqref{Ham-v} is similar to the Hamiltonian of the Nonlinear Schr\"odinger (NLS) equation, when expressed in the Madelung variables $(\rho, \phi)$ related to the complex wave function $\psi$ by $\psi = \sqrt{\eta}\exp(i\phi)$ and $v = \p_x \phi$ \cite{M1927}. 

Choosing $\kappa = \frac{i}{2}$ after transforming variables $\eta v={\rm Im}(\psi^*\p_x \psi)$ and $\eta=\|\psi\|^2$ in the Hamiltonian \eqref{Ham-v} yields the Hamiltonian for the focusing NLS equation with canonical complex variables $(\psi,\psi^*)$ and Planck constant $\hbar$ set equal to unity. 

Inserting $\kappa = \frac{i}{2}$ into the dispersion relation $\omega^2(k^2)$ in \eqref{disp-noalpha}, the dependence of frequency $\omega$ on wavenumber $k$ of the linearised solutions of the BKBK system is modified to
\begin{equation}
    \omega^2(k^2) = k^2\left(1 + \tfrac14 k^2\right)  \,. 
\end{equation}
This is a notable observation, since nonlinear modulation instability of water waves has been traditionally using with the focusing NLS equation, \cite{peregrine1983water}. Given this observation, studies of instability of water waves using the 1D BKBK system with real-valued $\kappa$ may be regarded as being complementary to the study of modulation instability of water waves using the NLS equation.

% Linear stability analysis seems to suggest
% \begin{align}
%     \begin{pmatrix}
%         i\omega & -ik\left(\kappa^2k^2/\eta_0+1\right) \\
%         -ik\eta_0 & i\omega
%     \end{pmatrix}
%     \begin{pmatrix}
%         \wh{\eta}\\ \wh{v}
%     \end{pmatrix} = 0
% \end{align}
% so that one have $w^2 = k^2\eta_0 + \kappa^2k^4$ for dispersion relation.

% where the state variables $(u,\eta)$ are identified as the scalar coefficients of $\wt{u} = u\,dx \in \Lambda^1(\mathbb{R})$ and $\wt{\eta} = \eta\,dx \in \operatorname{Den}(\mathbb{R})$. 
% % In this section, we will use the same notation $(u, \eta)$ for the $1$-forms and their scalar coefficients. When vector calculus notation are involved, $(u,\eta)$ are used to mean the scalar coefficients; when geometric operators such as the exterior derivatives and Lie derivatives are used, we use $(u,\eta)$ to denote the $1$-forms. Using this identification, the meaning of $u$ and $\eta$ should be clear from the context and cause no ambiguity.

% $\{\,\cdot\,,\, \cdot\,\}_c: C(\mcal{X})\times C(\mcal{X}) \rightarrow C(\mcal{X})$, $ \mcal{X} := \Lambda^1(\mathbb{R})\times \operatorname{Den}(\mathbb{R})$ associated with the constant-coefficient Hamiltonian structure is

\paragraph{Lie--Poisson bracket dynamics.}
Via an invertible change of variables $(u,\eta) \to (m:=\eta u,\eta)$ for $u\ne0$, the 1D BKBK system \eqref{BK-System-BAK3} may be cast into the following equivalent form
\begin{align}
    \begin{split}
        &\p_t m + \p_x \left(m v\right) + m\p_x v + \eta \p_x \mcal{B} = 0\,,\\
        &\p_t \eta + \p_x(v\eta) = 0\,, \label{BKBK-1D m eta}
    \end{split}
\end{align}
where the quantities $v$ and $\mcal{B}$ are defined by variations of the Hamiltonian functional \eqref{BKBK-Ham1} expressed in the new variables $(m,\eta)$, as
\begin{align}
\begin{split}
h (m,\eta) &:= 
\int_{\mathbb{R}} \frac{m^2}{2\eta} - \kappa m\, \p_x (\ln\eta)
+ \tfrac12\eta^2\, dx
\,, \quad\text{such that}\,,\\
\label{BKBK-Ham2}
v &:= \frac{\delta h}{\delta m} = u - \kappa \p_x \ln \eta\,,\quad \mcal{B} := \frac{\delta h}{\delta \eta} = -\frac{|u|^2}{2} + \frac{\kappa}{\eta}\p_x(\eta u) + \eta\,.
\end{split}
\end{align}
In fact, system \eqref{BKBK-1D m eta} takes the form of a Lie--Poisson system on the semidirect product Lie co-algebra $\mathfrak{s}^* = \mathfrak{X}^*(\mathbb{R}) \ltimes \operatorname{Den}(\mathbb{R})$ \cite{holm1983poisson,HoMaRa1998a,holm2025geometric}. In the Lie--Poisson formulation, one identifies $m$ as the scalar coefficient part of the momentum $1$-form density $\wt{m} := m\,dx \otimes dx \in \mathfrak{X}^*(\mathbb{R})$ and $\eta$ as the scalar coefficient part of the density $\wt{\eta} = \eta\,dx \in \operatorname{Den}(\mathbb{R})$. On the Lie co-algebra of the Lie algebra $\mathfrak{s} = \mathfrak{X}(\mathbb{R}) \ltimes \operatorname{Den}(\mathbb{R})$, the Lie--Poisson bracket $\{\cdot, \cdot\}$ may be defined in terms of the basis coefficients, for arbitrary $f, g$ functionals of $(m,\eta)$, as
\begin{align}
\begin{split}
    \{ f,g \} (m,\eta) &:=
\int_{\mathbb{R}}
m\left(\frac{\delta f}{\delta m}\p_x\frac{\delta g}{\delta m} - \frac{\delta g}{\delta m}\p_x\frac{\delta f}{\delta m}\right) + \eta\left(\frac{\delta f}{\delta m}\p_x\frac{\delta g}{\delta \eta} - \frac{\delta g}{\delta m}\p_x\frac{\delta f}{\delta \eta}\right)\,dx 
% &= \scp{m}{\left[\frac{\delta f}{\delta m}, \frac{\delta g}{\delta m}\right]}_{\mathfrak{X}(\mathbb{R})} + \scp{\eta}{\mcal{L}_{\frac{\delta f}{\delta m}}\frac{\delta g}{\delta \eta} - \mcal{L}_{\frac{\delta g}{\delta m}}\frac{\delta f}{\delta \eta}}_{\mcal{F}(\mathbb{R})}
\,.\label{LPB-def}
\end{split}
\end{align}
% where $[\,\cdot\, ,\,\cdot\, ]: \mathfrak{X}(\mathbb{R})\times \mathfrak{X}(\mathbb{R})\to \mathfrak{X}(\mathbb{R})$ is the vector field commutator and $\mcal{L}: \mathfrak{X}(\mathbb{R})\times \mcal{F}(\mathbb{R}) \rightarrow \mcal{F}(\mathbb{R})$ is the Lie derivative operator. The angle brackets, $\scp{\,\cdot\, }{\,\cdot\, }_{\mathfrak{X}(\mathbb{R})}$ and $\scp{\,\cdot\, }{\,\cdot\, }_{\mathfrak{X}(\mathbb{R})}$, denote the geometric duality pairing of $\mathfrak{X}(\mathbb{R})$ and $\mcal{F}(\mathbb{R})$, respectively.
% Introducing $\ad^* : \mathfrak{X}(\mathbb{R})\times \mathfrak{X}^*(\mathbb{R})\to \mathfrak{X}^*(\mathbb{R})$ as the dual operator to $[\cdot, \cdot]$ and the dual operator $\diamond : \operatorname{Den}(\mathbb{R})\times \mcal{F}(\mathbb{R}) \rightarrow \mathfrak{X}^*(\mathbb{R})$ defined by for all $(u, \eta, \phi)\in \mathfrak{X}(\mathbb{R}) \times \mcal{F}(\mathbb{R})\times \operatorname{Den}(\mathbb{R})$, 
% \begin{align}
%     \scp{-\phi \diamond \eta}{u}_{\mathfrak{X}(\mathbb{R})} = \scp{\phi}{\mcal{L}_u \eta}_{\mcal{F}(\mathbb{R})}\,,
% \end{align}
Thus, the 1D BKBK system \eqref{BKBK-1D m eta} can be expressed in Lie--Poisson operator form as
\begin{equation}
\p_t 
\begin{bmatrix}
m \\ \eta
\end{bmatrix}
= -
\begin{bmatrix}	
\p_x m + m\p_x  &	\eta\p_x	
\\
\p_x \eta	&	0
\end{bmatrix}
\begin{bmatrix}
\delta h/\delta m 
\\  
\delta h/\delta \eta
\end{bmatrix}
=- 
\begin{bmatrix}
\p_x \left(m \frac{\delta h}{\delta m}\right) + m\p_x \frac{\delta h}{\delta m} + \eta\p_x\frac{\delta h}{\delta \eta}
\\  
\p_x \big(\eta\tfrac{\delta h}{\delta m} \big)
\end{bmatrix}
.\label{BKBK-Poisson2} 
\end{equation}
Here, the spatial partial derivatives appearing in the Poisson matrix of \eqref{BKBK-Poisson2} each act on all terms to their right arising \emph{after} after performing the matrix vector multiplication. For background in the derivation and applications of semidirect-product Lie--Poisson brackets for continuum models, see, e.g., \cite{holm1983poisson,HoMaRa1998a,holm2025geometric}. 

Having expressed the BKBK system in the Lie--Poisson form in equation \eqref{BKBK-Poisson2}, one sees that the `dispersive term' proportional to the constant $\kappa$, modifies the transport velocity ($\delta h/ \delta m$) from $u$ in the classical shallow water case to $\delta h/ \delta m = v:= u - \kappa \p_x\ln \eta$ in the dynamics for both depth $\eta$ and momentum density $m$.

\paragraph{The equivalent Euler--Poincar\'e derivation.}
The Lie--Poisson Hamiltonian form of the BKBK system has an equivalent Euler--Poincar\'e Lagrangian form arising from the corresponding Hamilton variational principle. Namely, the 1D BKBK system \eqref{BK-System-BAK3} can be derived from an Euler--Poincar\'e variational principle \cite{HoMaRa1998a} with advected quantity $\eta$ given by, 
\begin{align}\label{1D BKBK lag}
    0 = \delta S = \delta \int_{t_0}^{t_1}\ell(v, \eta)\,dt = \delta \int_{t_0}^{t_1} \int_{\mathbb{R}} \frac{1}{2}\eta |v + \kappa \p_x \ln\eta|^2 - \frac{1}{2}\eta^2\,dx\,dt\,,
\end{align}
with constrained variations given by
\begin{align}\label{1D BKBK vars}
\delta v = \p_t w - w\partial_x v + v\partial_x w
\quad\hbox{and}\quad
\delta \eta = - \p_x(w \eta)
\,.
\end{align}
Here, $w$ is an arbitrary vector field that vanishes on the temporal endpoints and obeys appropriate decay conditions as $|x|\rightarrow \infty$. After taking these constrained variations and integrating by parts in space and time the Euler--Poincar\'e formulation with Lagrangian \eqref{1D BKBK lag} recovers the 1D BKBK system in \eqref{BK-System-BAK3}. For the intrinsic expression of the Euler--Poincar\'e and Lie--Poisson formulation of the BKBK system in terms of actions of the diffeomorphism group, see, e.g., \cite{HoMaRa1998a}.

A virtue of having derived the 1D BKBK system \eqref{BK-System-BAK3} via an Euler--Poincar\'e variational principle with advected quantities is that the definition of its Lagrangian extends easily to any number of spatial dimensions. 
In particular, one may derive the 2D BKBK system from the same variational principle, as shown in section \ref{sec-4}. As we will see, equivalence of the Euler--Poincar\'e variational principle and the Lie--Poisson Hamiltonian structure implies that the 2D BKBK system also satisfies the semidirect-product Lie--Poisson bracket formulation that appears naturally in the ideal dynamics of all continuum models. Moreover, the Madelung transform of the 1D BKBK system for $\kappa=i/2$ to the focusing NLS equation also applies for the 2D BKBK system, but only in the case of potential flow, 
 \begin{align}
 \bu=\nabla \phi = \bv + \kappa \nabla \ln\eta 
 \,.
 \label{2D-potentflow}
 \end{align}

\section{BKBK system in 2D}\label{sec-3}
%%%%%%%%%%%%%%%%%%%%%%%%%%%%%%%%%%%%%%%%%%%%%%%%%%%%%%%%
This section extends the geometric structures of the 1D BKBK system to 2D. In particular, the Euler--Poincar\'e derivations of the 2D BKBK system, as well as its Hamiltonian structures and their implications such as the existence of potential vorticity and its associated conservation laws as presented.
%%%%%%%%%%%%%%%%%%%%%%%%%%%%%%%%%%%%%%%%%%%%%%%%%%%%%%%%

\paragraph{Variational principles of 2D BKBK.}
To extend the 1D BKBK system \eqref{BK-System-BAK3} defined on the real line to the 2D plane, we denote by $\bu = \bu(\bx,t)$ for the fluid velocity and $\eta = \eta(\bx,t)$ for the fluid depth where $\bx \in \mathbb{R}^2$. Geometrically, we have the fluid velocity vector field $u = \bu\cdot \,\nabla \in\mathfrak{X}(\mathbb{R}^2)$ and the fluid depth as the density $\wt{\eta} = \eta\, d^2x \in \operatorname{Den}(\mathbb{R}^2)$. The BKBK modified transport velocity vector field is defined to be,
\begin{align}
\bv := \bu - \kappa \nabla \ln\eta 
\,.\label{DSW-transport-vel}
\end{align}
%
%\footnote{Here, to simplify the notation, the operation $\nabla$ is treated as having either contravariant or covariant indices, depending upon its usage.}
We postulate that the dynamics of $\bu$ and $\eta$ are obtained from the following Euler--Poincar\'e variational principle with advected quantity,  
\begin{align}\label{2D BKBK lag}
    0 = \delta S[\bv,\eta] = \delta \int_{t_0}^{t_1}\ell(\bv, \eta)\,dt = \delta \int_{t_0}^{t_1} \int_{\mathbb{R}^2} \frac{1}{2}\eta |\bv + \kappa \nabla \ln\eta|^2 - \frac{1}{2}\eta^2\,d^2x\,dt\,,
\end{align}
with constrained variations $\delta \bv = \p_t \bw - [\bw, \bv]$ and $\delta \eta = - \nabla\cdot(\bw \eta)$, where $\bw$ is an arbitrary vector that vanish on temporal boundaries and appropriate decay conditions as $x\rightarrow \infty$. Here, $[\cdot,\cdot]$ is the Jacobi--Lie bracket of vector fields, $[\bw, \bv] = \bw \cdot\nabla \bv - \bv\cdot \nabla \bw$.
Applying the Euler--Poincar\'e theorem \cite{HoMaRa1998a}, the resulting Euler--Poincar\'e equation and the advection equation of $\eta$ form the 2D BKBK system
\begin{align}
    \begin{split}
        &\p_t \bu + \bv\cdot \nabla \bu + u_i\nabla v^i = \nabla {\cal B}\,,\\
        &\p_t \eta + \operatorname{div}(\eta \bv ) = 0\,, \\
       & \text{where} \quad  \bv = \bu - \kappa \nabla \ln\eta\,, \quad {\cal B}:= \left(\frac{1}{2}|u|^2 - \frac{\kappa}{\eta}\operatorname{div}\left(\eta \bu\right) - \eta\right)
    \,.\end{split}\label{2D BKBK}
\end{align}
In the 2D BKBK system \eqref{2D BKBK}, one sees that the term in the Lagrangian functional \eqref{2D BKBK lag} proportional to $\kappa$ plays two roles. First, it serves to add the standard dispersion term for shallow water waves, $-\kappa \eta^{-1}\nabla \operatorname{div}(\eta \bu )$. Second, it serves to enhance the shallow water fluid transport velocity $\mathbf{u}$ by an added transport velocity, $- \kappa \nabla \ln\eta$, which involves the gradient of the wave elevation.

By a direct calculation, the 2D BKBK system in \eqref{2D BKBK} implies the following Kelvin theorem 
\begin{align}
\frac{d}{dt}\oint_{c(v)}\mathbf{u}\cdot d\bx 
= 
\oint_{c(v)}\big(\p_t \mathbf{u} - \mathbf{v} \times {\rm curl} \mathbf{u} \big) \cdot d\bx 
= 
\oint_{c(v)} \nabla \tilde{\cal B} \cdot d\bx = 0
\,,\label{DSW-Kelvin-fluid}
\end{align}
where $c(v)$ is a material loop moving with the BKBK modified transport velocity vector field $\bv$. 
\begin{remark}
The definition $\bv = \bu - \kappa \nabla \ln \eta$ implies a constrained variation for $\bu$ arising  from the constrained variations of $\bv$ and $\eta$. In particular, 
\begin{align}
\begin{split}
    \delta \bu = \delta \bv + \kappa \nabla \left(\frac{1}{\eta}\delta \eta\right) &= \p_t \bw - [\bw, \bv] - \kappa\nabla\left(\frac{1}{\eta} \nabla\cdot(\bw\eta)\right)\\
    & = \p_t \bw - \ [\bw, \bu - \kappa \nabla \ln \eta] - \kappa\nabla\left(\frac{1}{\eta} \nabla\cdot(\bw\eta)\right)\,.
\end{split}
\end{align}
Expressing the action principle \eqref{2D BKBK lag} in terms of $(\bu,\eta)$ and imposing the constrained variations for $\eta$ and $\bu$ as presented above yield the 2D BKBK system \eqref{2D BKBK}.
\end{remark}
\begin{remark}
    One may express the 2D BKBK system \eqref{2D BKBK} in the variables $(\bv ,\eta)$ as
    \begin{align}
        \begin{split}
            &\p_t \bv + \bv\cdot \nabla \bv + v_i \nabla v^i = \nabla \wh{\mcal{B}}\,, \quad \text{where} \quad \wh{\cal{B}} = \frac{|\bv|^2}{2} + \frac{\kappa^2}{2}|\nabla \ln \eta|^2 - \frac{\kappa^2}{\eta}\triangle \eta - \eta\,, \\
            & \p_t \eta + \nabla\cdot\left(\bv \eta\right) = 0\,.
        \end{split}
    \end{align}
    The Kelvin circulation dynamics for the transport velocity $\bv$ has the same form as for the fluid velocity $\bu$. Let $c(v)$ be the material loop moving with the vector field $\bv$,
    \begin{align}
    \frac{d}{dt}\oint_{c(v)}\mathbf{v}\cdot d\bx 
    = 
    \oint_{c(v)} \nabla \wh{\cal B} \cdot d\bx = 0\,.
    \end{align}
    % As we shall see in the stability analysis, the BKBK system in terms of $(\bv,\eta)$ variables have the same stability properties as the classical shallow water equations.
    % % \todo[inline]{RH: Darryl, the Bernoulli function $\wh{\mcal{B}}$ closely resembles the Bernoulli function of nonlinear Schr\"oedinger's equation in Madelung variables.
    % % What we can say about this observation?\\
    % % {\color{red}DH: I guess we can write out the formulas demonstrating your observation, recall the Hasimoto transformation and say it suggests further investigation, for later studies.}}
\end{remark}

\paragraph{Hamiltonian structures for 2D BKBK.}
In generalising the semidirect-product Lie--Poisson bracket in equation \eqref{LPB-def} from one spatial dimension to two spatial dimensions, the Lie--Poisson Hamiltonian formulation of the 1D BKBK system in \eqref{BKBK-Poisson2} generalises to the following 2D BKBK system in $\mathbb{R}^2$ index notation: 
\begin{equation}
\p_t 
\begin{bmatrix}
m_i \\ \eta
\end{bmatrix}
= -
\begin{bmatrix}	
\p_j m_i + m_j\p_i  &	\eta\p_i	
\\
\p_j \eta	&	0
\end{bmatrix}
\begin{bmatrix}
\delta h/\delta m_j
\\  
\delta h/\delta \eta
\end{bmatrix}
, \label{BKBK-Lie-Poisson2D} 
\end{equation}
with $i,j=1,2$ and standard Einstein summation notation. Its Hamiltonian functional is written in the variables $(\mathbf{m},\eta)$ as 
\begin{align}
 \begin{split}
    h (\bm,\eta) &:=  \int_{\mathbb{R}} \frac{|\bm|^2}{2\eta} - \kappa \bm\cdot \nabla (\ln\eta) + \frac12\eta^2\, d^2x \,,\quad \text{such that}\\
    \frac{\delta h}{\delta \bm} &= \bu - \kappa \nabla \ln\eta = \bv\,,\quad \frac{\delta h}{\delta \eta} = - \frac{|\bu|^2}{2} + \eta + \frac{\kappa}{\eta}\nabla\cdot(\eta \bu) = - {\cal B} \,.\label{BKBK-Ham2D}
    \end{split}
\end{align}
One may directly calculate the inverse transformation of variables $(\bm,\eta)\to (\bu=\bm/\eta,\eta)$ for the 2D version of the Lie--Poisson structure in \eqref{BKBK-Lie-Poisson2D} and thereby determine the equations of motion in the standard shallow water variables $(\bu,\eta)$. In this case, the Poisson matrix operated may be computed as follows
\begin{align}
    \begin{bmatrix}
\delta_{ki}/\eta & - m_k /\eta^2		
\\							
0	&	1		
\end{bmatrix}
\begin{bmatrix}	
\p_j m_i + m_j\p_i  &	\eta\p_i	
\\
\p_j \eta	&	0
\end{bmatrix}
\begin{bmatrix}
\delta_{jl}/\eta	&	0		
\\							
- m_l / \eta^2   &	1		
\end{bmatrix}
=
\begin{bmatrix}
(u_{k,l} - u_{l,k} )/\eta &	\p_k	
\\							
\p_l & 0		
\end{bmatrix}
\,. \label{eq:PV bracket}
\end{align}
Let $q \in \mcal{F}(\mathbb{R}^2)$ be the scalar potential vorticity (PV) defined by 
\begin{align}
    q:= \frac{1}{\eta} \mathbf{\wh{z}}\cdot\operatorname{curl}\bu = \frac{1}{\eta}(\p_x u_2 - \p_y u_1) \,.
\label{DSW-def-PV} 
\end{align}
For notational convenience, the 2D velocity $\bu$ in \eqref{DSW-def-PV} is written as a three dimensional vector $\bu = (u_1,u_2,0)$ and the function $q$ for PV is written as a vertical component $\mathbf{q} := q \mathbf{\wh{z}}$, so that $\mathbf{q}\times \bu$ rotates a horizontal vector $\bu$ clockwise by $\pi/2$. 
The transformed Poisson matrix in \eqref{eq:PV bracket} yields the following Poisson bracket
\begin{align}
    \{F,H\}(\bu,\eta) = \int_{\mathbb{R}}\left[ -\frac{\delta F}{\delta \bu}\cdot q \frac{\delta H}{\delta \bu}^T + \frac{\delta G}{\delta \eta} \nabla \cdot \frac{\delta F}{\delta \bu} - \frac{\delta F}{\delta \eta} \nabla \cdot \frac{\delta G}{\delta \bu} \right]\,dx  
    \,.\label{PB-def} 
\end{align}
Here, the superscript $(\cdot)^T$ is defined for two dimensional vectors as $(u_1, u_2)^T := (-u_2, u_1)$.
Consequently, the 2D BKBK system \eqref{2D BKBK} may be expressed with a Poisson operator \eqref{eq:PV bracket} involving the potential vorticity $q$ in the vector form $\mathbf{q} := q \mathbf{\wh{z}}$ as 
\begin{align}
\begin{split}
\p_t 
\begin{bmatrix}
\bu \\ \eta
\end{bmatrix}
&= -
\begin{bmatrix}
\mathbf{q}\times & \nabla \\							
\nabla \cdot & 0		
\end{bmatrix}
\begin{bmatrix}
\delta \tilde{h}/\delta \bu = \eta \bv 
\\
\delta \tilde{h}/\delta \eta = \tilde{\cal B}
\end{bmatrix}
\,,\end{split}
\label{LP-DSW-System2}
\end{align} 
for Hamiltonian \eqref{BKBK-Ham2D} written in the variables $(\mathbf{u},\eta)$ as
\begin{align}
\begin{split}
    \wt{h} (\mathbf{u},\eta) &:= \int_M \eta \frac{|\mathbf{u}|^2}{2} - \kappa\mathbf{u} \cdot\nabla\eta + \frac{1}{2} \eta^2\, d^2x \,, \quad \text{such that}\\
    \frac{\delta \wt{h}}{\delta u} &= \eta\mathbf{u} - \kappa \nabla \eta = \eta\bv\,, \quad \frac{\delta \wt{h}}{\delta \eta} = \frac12 |\mathbf{u}|^2 + \kappa\, {\rm div}(\mathbf{u}) +  \eta =:\tilde{\cal B}(\mathbf{u},\eta)\,.
\end{split}
\label{DSW-Ham-u}
\end{align} 

The PV scalar function $q$ satisfies an advection equation which may be readily obtained by taking the $\operatorname{curl}$ of the $\bu$ equation in \eqref{2D BKBK}
\begin{align}
\Big(\p_t + \mathbf{v} \cdot \nabla\Big)q =
\Big(\p_t + (\mathbf{u} - \kappa \nabla \ln\eta)\cdot \nabla\Big)q = 0
\,.\label{DSW-omega-eqn-fluid}
\end{align}
In turn, the PV advection equation \eqref{DSW-omega-eqn-fluid} implies conservation by the 2D BKBK equations in \eqref{LP-DSW-System2} of the quantity 
\begin{align}
C_\Phi := \int_{\mathbb{R}^2} \eta \Phi(q)\,d^2x
\,,\label{DSW-Casimir}
\end{align}
for any differentiable function $\Phi$. It is clear that the functional $C_\Phi$ is conserved for any differentiable function $\Phi$, since the depth $\eta$ and the potential vorticity $q$ are both advected quantities. Moreover, the variational derivative of the functional $C_\Phi$ comprises a null eigenvector of the Lie--Poisson operator in terms of variables $(m,\eta)$ in equation \eqref{BKBK-Lie-Poisson2D}.  Consequently, the functionals $C_\Phi$ in \eqref{DSW-Casimir} for the Lie--Poisson formulation of the 2D BKBK system would in fact be conserved for \emph{any} Hamiltonian depending on the variables $(m,\eta)$. 

\section{Stability of equilibrium solutions of the 2D BKBK system}\label{sec-4}
%%%%%%%%%%%%%%%%%%%%%%%%%%%%%%%%%%%%%%%%%%%%%%%%%%%%%%%%
This section derives equilibrium conditions for the 2D BKBK system following the nonlinear stability analysis of \cite{holm1985nonlinear}. This is done in two stages. First, we consider the critical points of the sum $h_\Phi(\mathbf{u},\eta) := 
h(\mathbf{u},\eta) + C_\Phi $ of the Hamiltonian found in section \ref{sec-3} plus its Casimir constants of motion, $C_\Phi$ in \eqref{DSW-Casimir}. Then, by demanding positivity of the second variation of $h_\Phi(\mathbf{u},\eta)$, we derive spectral conditions for linear Lyapunov stability of the corresponding energy-Casimir class of equilibrium solutions of the 2D BKBK system.

%%%%%%%%%%%%%%%%%%%%%%%%%%%%%%%%%%%%%%%%%%%%%%%%%%%%%%%%

Let $(\bu_e, \eta_e, q_e)$ be an equilibrium solution to the 2D BKBK system in \eqref{2D BKBK} augmented with the advection equation of scalar PV, denoted as $q$. We denote the equilibrium transport velocity vector $\bv_e$ by $\mathbf{v}_e = \mathbf{u}_e - \kappa \nabla \ln\eta_e$. The equilibrium solutions satisfy
\begin{align}
\mathbf{v}_e \times \omega_e \mathbf{\wh{z}} = \nabla\tilde{\cal B}(\mathbf{u}_e,\eta_e)
\,,\quad 
\mathbf{v}_e \cdot \nabla q_e = 0
\,,\quad\hbox{and}\quad
{\rm div}(\eta_e \mathbf{v}_e) = 0
\,.
\label{BK-StationCond}
\end{align}
Consequently, at equilibrium the gradients $\nabla\tilde{\cal B}_e$ and $\nabla q_e$ are collinear because they are both orthogonal in the $\mathbb{R}^2$ plane to $\bv_e$. 
A sufficient condition for this collinearity is a functional relationship which we choose to write as
\begin{align}
\tilde{\cal B}(\mathbf{u}_e,\eta_e) =  q_e \Phi'(q_e) -  \Phi(q_e) 
\,,\label{BK-B-Cond1}
\end{align}
for some twice differentiable function $\Phi(\xi)$, with $\xi \in \mathbb{R}$, defined wherever $\nabla q_e$ does not vanish.  

Applying the operator $q_e^{-1}\mathbf{\wh{z}}\times$ to the first equilibrium condition in equation \eqref{BK-StationCond} 
and using relation \eqref{BK-B-Cond1} yields 
\begin{align}
\eta_e \mathbf{v}_e  = q_e^{-1}\mathbf{\wh{z}} \times \nabla \tilde{\cal B}(\mathbf{u}_e,\eta_e) 
= q_e^{-1}\mathbf{\wh{z}} \times \nabla \big(q_e \Phi'(q_e) -  \Phi(q_e) \big) 
= \mathbf{\wh{z}} \times \nabla \Phi'(q_e)
\,,\label{BK-B-Cond2}
\end{align}
which also implies the last equilibrium condition in equation \eqref{BK-StationCond}; namely,  ${\rm div}(\eta_e \mathbf{v}_e) = 0$. 

\paragraph{The $1^{st}$ variation of $h_\Phi(\mathbf{u},\eta)$}

\begin{proposition}
Critical points of the sum of the 2D BKBK Hamiltonian and its Lie--Poisson Casimirs given by 
\begin{align}
\begin{split}
h_\Phi(\mathbf{u},\eta) := 
h(\mathbf{u},\eta) + C_\Phi 
:= 
\int_{\mathbb{R}^2} \eta \frac{|\mathbf{u}|^2}{2} 
- \kappa \mathbf{u} \cdot\nabla \eta
+ \frac{1}{2} \eta^2 + \eta \Phi(q)
\,d^2x 
\,,
\end{split}
\label{BKBK-Ham1 2D}
\end{align} 
are equilibrium solutions of the 2D BKBK system in \eqref{LP-DSW-System2}.
\end{proposition}

\begin{proof}
Upon using the definition $q:=\eta^{-1} \mathbf{\hat{z}} \cdot {\rm curl}\, \mathbf{ u}$ in \eqref{DSW-def-PV} to obtain the variational relation
\begin{align}
\delta q = -\, \eta^{-1} q \,\delta \eta + \eta^{-1} \mathbf{\hat{z}} \cdot {\rm curl}\, \delta \mathbf{ u}
\,,\label{Ham-q[var1}
\end{align}
and assuming that the velocity variation $\delta \mathbf{u}$ vanishes at spatial infinity of $\mathbb{R}^2$ (or has zero circulation on any boundary which may be present in an enclosed flow domain ${\cal D}$), then one finds that critical points of $h_\Phi(\mathbf{u},\eta)$ in \eqref{BKBK-Ham1} satisfy the following equations,
\begin{align}
\begin{split}
0 = \delta h_\Phi (\mathbf{u},\eta) 
&= \scp{ \eta \mathbf{u} - \kappa\, \nabla \eta - \mathbf{\hat{z}}\times \nabla \Phi'(q) } {\delta \mathbf{u}} 
\\ & \quad + \scp{ \tfrac12 |\mathbf{u}|^2 +  {\kappa}{\rm div}(  \mathbf{u} ) + \eta + \Phi(q) - q \Phi'(q)}{ \delta \eta } 
\\ & = \scp{ \eta \mathbf{v} - \mathbf{\hat{z}}\times \nabla \Phi'(q) } {\delta \mathbf{u}} 
+ \scp{ \tilde{\cal B}(\mathbf{u},\eta) + \Phi(q) - q \Phi'(q)}{ \delta \eta } 
\,,\label{Ham-var1}
\end{split}
\end{align}
in which the angle brackets $\scp{\,\cdot\, }{\,\cdot\, }$ denote $L^2$ pairing on $\mathbb{R}^2$ (or an enclosed flow domain ${\cal D}$).
\end{proof}

Thus, critical points of $h_\Phi (\mathbf{u},\eta)$ in \eqref{BKBK-Ham1} are equilibrium solutions of the 2D BKBK equations in \eqref{LP-DSW-System2}. These critical points satisfy the conditions
\begin{align}
\eta_e \mathbf{v}_e -  \mathbf{\hat{z}}\times \nabla \Phi'(q_e) = 0 
\quad\hbox{and}\quad
\tfrac12 |\mathbf{u}_e|^2 +  \kappa{\rm div}( \mathbf{u}_e ) + \eta_e + \Phi(q_e) - q_e \Phi'(q_e) = 0
\,,\label{Equil-conds}
\end{align}
for any choice of the functions $\Phi(q)$ of potential vorticity $q$ in \eqref{DSW-Casimir}. 
Consequently, the twice-differentiable functions $\Phi(q)$ provide a class of equilibrium solutions of the 2D BKBK system in \eqref{LP-DSW-System2}.

\paragraph{The 2nd variation of $h_\Phi(\mathbf{u},\eta)$}

Upon recalling the definitions of  $\tilde{\cal B}$ in \eqref{DSW-Ham-u} and $\mathbf{v}$ in \eqref{DSW-transport-vel},
\begin{align}
\tilde{\cal B}(\mathbf{u},\eta)
:= \tfrac12 |\mathbf{u}|^2 + g \eta + \kappa {\rm div}(\mathbf{u})  
\quad\hbox{and}\quad
\mathbf{v} :=  \mathbf{u} - \kappa \nabla \ln\eta
\,,\label{DSW-B-tilde-redux}
\end{align} 
one may rewrite the variation $\delta h_\Phi (\mathbf{u},\eta) $ in \eqref{Ham-var1} to separate its $\kappa$ terms as 
\begin{align}
\begin{split}
\delta h_\Phi (\mathbf{u},\eta) 
&= \scp{ \eta \mathbf{u}  - \mathbf{\hat{z}}\times \nabla \Phi'(q) } {\delta \mathbf{u}}  
\\ & \quad + \scp{ \tfrac12 |\mathbf{u}|^2  + \eta + \Phi(q) - q \Phi'(q)}{ \delta \eta } 
\\ & \quad  + \kappa \scp{{\rm div}( \mathbf{u} ) }{\delta \eta}
\,.\label{Ham-var1-redux}
\end{split}
\end{align}
One may then calculate the second variation $\delta^2 h_\Phi (\mathbf{u}_e,\eta_e) $ which is in fact the conserved energy for the linearised dynamics of perturbations of the BKBK equilibria arising as critical points of $h_\Phi (\mathbf{u}_e,\eta_e)$ in \eqref{BKBK-Ham1}.
The calculation of the second variation $\delta^2 h_\Phi (\mathbf{u}_e,\eta_e)$ then yields
\begin{align}
    \begin{split}
        \delta^2 h_{\Phi}(\bu_e,\eta_e) = \int_{\mathbb{R}^2}2\delta \eta\, \delta\bu\cdot \bu_e + \eta_e (\delta\bu)^2 + (\delta \eta)^2 - 2\kappa \delta \bu\cdot \delta (\nabla \eta) + \eta_e\Phi''(q_e)(\delta q)^2\,d^2x
    \,.\end{split}
\end{align}
By treating $\delta \eta$ and $\delta \nabla \eta$ as independent variables, one may organise $\delta^2 h_\Phi (\mathbf{u}_e,\eta_e)$ as
\begin{align}
    \delta^2 h_\Phi (\mathbf{u}_e,\eta_e) 
&=
\int_{\mathbb{R}} 
\begin{bmatrix}
\delta \mathbf{u} \\ \delta \eta \\ \delta \nabla \eta
\end{bmatrix}^T
\begin{bmatrix}
\eta_e & u_e & -\kappa
\\
u_e   & 1  &0 \\
-\kappa &0 &0
\end{bmatrix}
\begin{bmatrix}
\delta \mathbf{u} \\ \delta \eta \\ \delta \nabla \eta
\end{bmatrix}
+
\eta_e \Phi''(q_e)(\delta q)^2
d^2x\,. \label{2nd-var-conds-2D 1}
\end{align}
The linear stability of the 2D BKBK equilibrium solutions in the class $\delta h_{\Phi}(\bu_e,\eta_e) = 0$ requires the symmetric operator in \eqref{2nd-var-conds-2D 1} to be positive definite. Via direct calculation, the determinant of the symmetric operator, $\operatorname{\det}(\delta^2 h_{\Phi}) = -\kappa^2 < 0 $ such that the linear Lyapunov stability conditions is not satisfied. Enforcing the condition $\delta (\nabla\eta) = \nabla \delta \eta$, the extent to which linear Lyapunov stability condition is not satisfied is revealed by re-organising $\delta^2 h_\Phi (\mathbf{u}_e,\eta_e)$ as 
\begin{align}
\begin{split}
\delta^2 h_\Phi (\mathbf{u}_e,\eta_e) 
&=
\int_{\mathbb{R}} 
\begin{bmatrix}
\delta \mathbf{u} \\ \delta \eta
\end{bmatrix}^T
\begin{bmatrix}
\eta_e & u_e - \kappa\nabla
\\
u_e + \kappa\nabla\cdot  & 1  
\end{bmatrix}
\begin{bmatrix}
\delta \mathbf{u} \\ \delta \eta 
\end{bmatrix}
+
\eta_e \Phi''(q_e)(\delta q)^2
d^2x\,.\label{2nd-var-conds-2D}
\end{split}
\end{align}
For the case of shallow water dynamics arising for $\kappa=0$, requiring definiteness in sign of $\delta^2 h_\Phi (\mathbf{u}_e,\eta_e) $ provides sufficient
conditions for linear Lyapunov stability of 2D shallow water equilibrium solutions that are obtained from setting $\delta h_\Phi (\mathbf{u}_e,\eta_e) =0$ with $\kappa =0$. This was shown for similar equations for barotropic ideal compressible fluid dynamics in \cite{holm1983nonlinear,holm1985nonlinear}. For the singular perturbation of nonzero $\kappa$, though, linear stability of perturbed 2D BKBK equilibrium solutions in the class $\delta h_\Phi (\mathbf{u}_e,\eta_e) =0$ require in addition that the spectrum of the symmetric elliptic operator in \eqref{2nd-var-conds-2D} be positive definite. The symbol of the elliptic operator in \eqref{2nd-var-conds-2D} is of the form
\begin{align}
\sigma(|\mathbf{k}|^2) \approx \eta_e - \big( \mathbf{u}_e^2 + \kappa^2\,|\mathbf{k}|^2 \big) 
\,,\label{spectral-conds}
\end{align}
in which the equilibrium solutions $\eta_e$ and $\mathbf{u}_e$ satisfy ${\eta}_e(\bx)-\mathbf{{u}}_e(\bx)^2> 0$. Hence, the linear Lyapunov stability conditions obtained from requiring positivity of the second variation of the Hamiltonian $\delta^2 h_\Phi (\mathbf{u}_e,\eta_e)$ in \eqref{2nd-var-conds-2D} cannot hold for arbitrarily high wavenumber magnitudes $|\mathbf{k}|$.

%%%%%%%%%%%%%%%%%%%%%%%%%%%%%%%%%%%%%%%%%%%%%%%

\section{Simulations of solution behaviour for the BKBK equations}\label{sec-5}
%%%%%%%%%%%%%%%%%%%%%%%%%%%%%%%%%%%%%%%%%%%%%%%%%%%%%%%%
This section investigates computational simulations of the BKBK system in 1D and 2D. In the 1D simulations of the BKBK system \eqref{BK-System-BAK3}, high-wavenumber instabilities are controlled here by introducing 4th-order dissipation in both velocity and elevation, reminiscent of the Kuramoto--Sivashinsky equation. Computational simulations are then used to investigate the travelling wave solution and nonlinear wave interactions of initially Gaussian wave elevations. In the 2D simulations of the BKBK system \eqref{2D BKBK}, we employ a Hamiltonian regularisation where the 2D BKBK Hamiltonian \eqref{BKBK-Ham2D} is augmented to include an energy penalty for wave height gradients. Using the regularised BKBK equations \eqref{2DBKBK-system}, we investigate the interaction of waves with and without potential vorticity, for both signs of $\kappa$.
%%%%%%%%%%%%%%%%%%%%%%%%%%%%%%%%%%%%%%%%%%%%%%%%%%%%%%%%

\subsection{1D BKBK system simulation}
As discussed upon introducing the 1D BKBK system in section \ref{sec-BKBK-1D}, the dispersion relation for BKBK in equation \eqref{disp-noalpha} indicates high-wavenumber instability. Direct simulation for \eqref{BK-System-BAK3} is unstable even for analytical travelling solutions. This is reminiscent of the similar negative diffusion situation in the well-known Kuramoto-Shivashinski equation in 1D, 
\begin{align}
u_t+u u_x+u_{x x}+u_{x x x x}=0
\,.\label{KS-eqn} 
\end{align}
For periodic boundary conditions, we introduce the following 4th-order dissipation term into the motion equation of the BKBK 1D system \eqref{BK-System-BAK3} as,
\begin{align}\begin{split}
{u}_t &= - \left({u}^2 / 2+{\eta}+\kappa {u}_x+\nu u_{xxx}\right)_x 
\,,\\
{\eta}_t &= - \left({u} {\eta}-\kappa {\eta}_x + \nu \eta_{xxx}\right)_x
\,.\label{BK-System-BAK3-reg2} 
\end{split}
\end{align} 
When the fourth-order dissipation coefficients are unequal, $\nu_u \neq \nu_\eta$, the dispersion relation acquires a $k^6$ cross term proportional to $\kappa\left(\nu_u-\nu_\eta\right) k^6$ in the quadratic form for $-i \omega$. This term couples phase and damping and can create a finite band of unstable wavenumbers even when both $\nu_u, \nu_\eta>0$. The symmetric choice of $\nu$ as in \eqref{BK-System-BAK3-reg2} leads to the following dispersion relation
$$
\omega/k = -i\nu k^3 \pm \sqrt{g\eta_0-\kappa^2 k^2}
\,,\quad\hbox{with}\quad g\eta_0 = 1
\,,\label{disperse-relation}
$$
which is the same as \eqref{disp-noalpha} with an additional pure damping term and two dimension-free numbers $g\eta_0/(k \kappa)^2$ and $k^2 \nu/\kappa$. Here the terms $k \kappa$, $\sqrt{g\eta_0}$, $\nu k^3$, $\omega/k$ all have the units of velocity, $[u]=[L]/[T]$. For wavenumbers $|k|>k_c=\sqrt{g \eta_0} /|\kappa|$, choosing $\nu>\nu_{cr}=\frac{2|\kappa|^3}{3 \sqrt{3} g \eta_0}$ is necessary to ensure the growth rate satisfies $\operatorname{Im} \omega(k)<0$ for all $k$, thereby eliminating the unstable band of modes. 

\begin{figure}[H]
    \centering
    \begin{subfigure}[t]{0.45\linewidth}
        \centering
        \includegraphics[width=\linewidth]{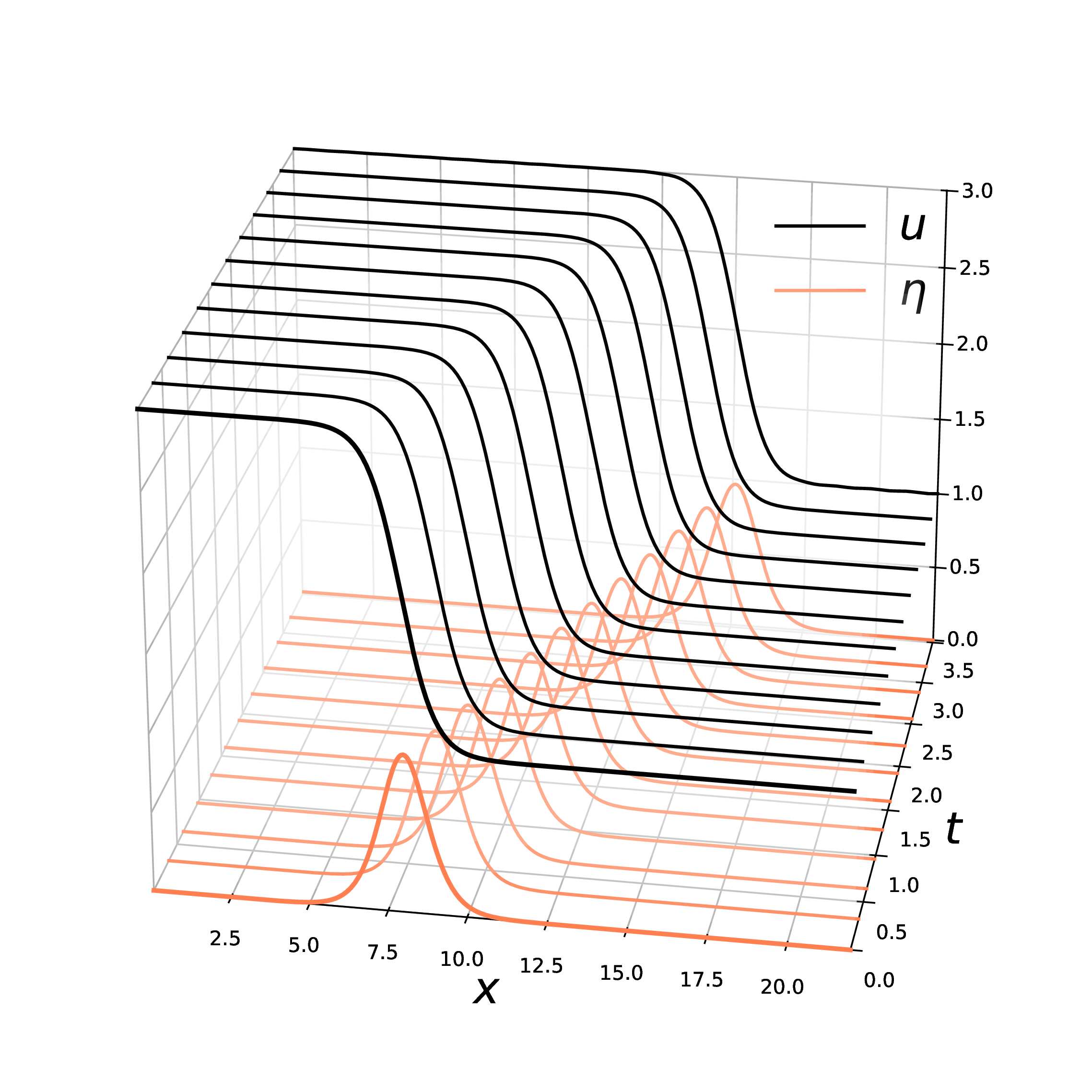}
        \caption{1D BKBK travelling wave \eqref{eq: bkbk_tw} preserves its structure for $\nu>0$. $\kappa = 0.5, c=2, \lambda=2,  \nu=0.01 $}
        \label{fig: BKBK_tw}
    \end{subfigure}
    \hfill
    \begin{subfigure}[t]{0.45\linewidth}
        \centering
        \includegraphics[width=\linewidth]{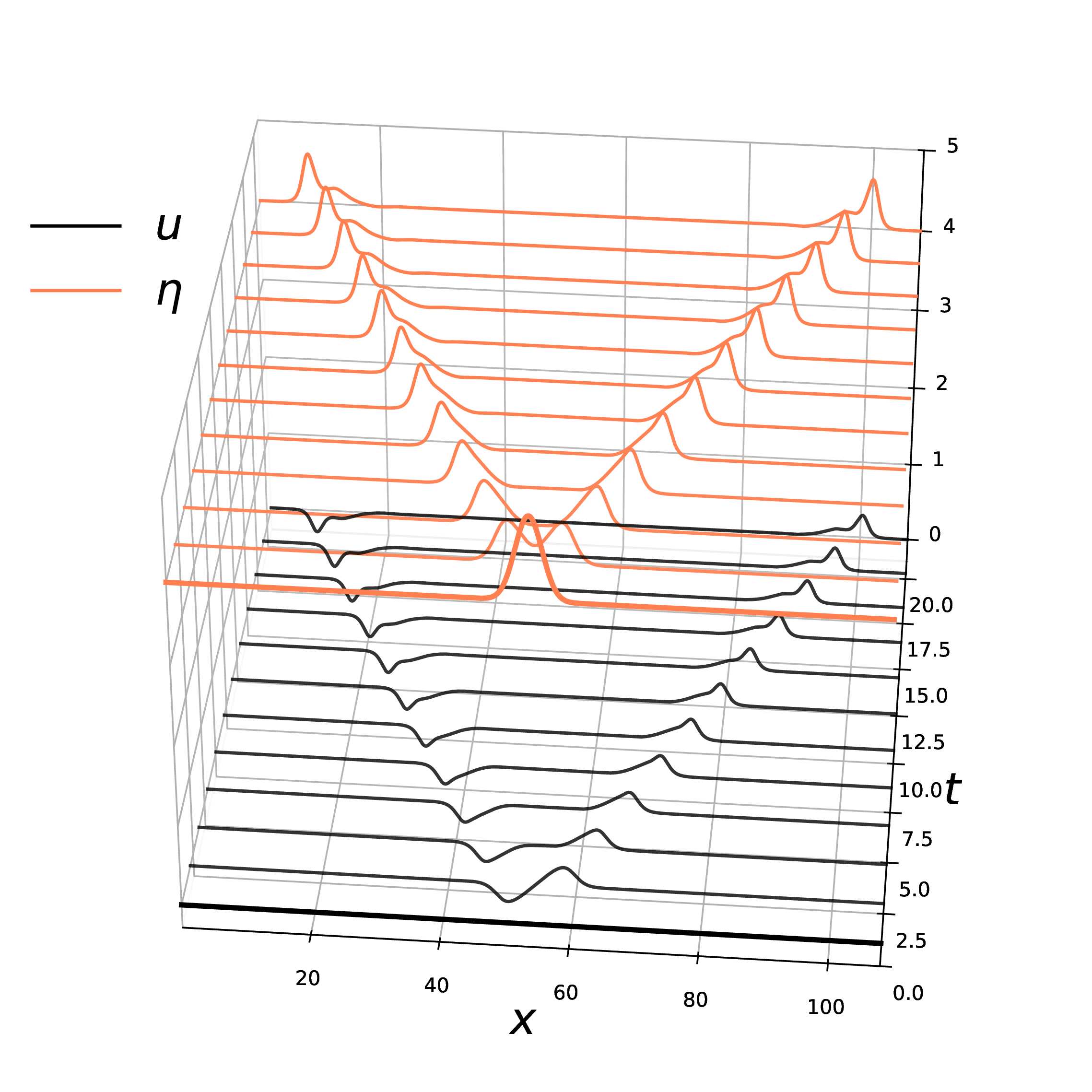}
        \caption{Two counter-propagating pulses generated from a Gaussian depth initial condition: $u=0$, $\eta=4+\exp(-(x-54)^2/8)$, $\kappa=-0.5, \nu=0.01$. }
        \label{fig: BKBK_gaussian}
    \end{subfigure}
    \caption{Space–time evolution of the velocity $u$ (black) and depth $\eta$ (orange) for the 1D BKBK system regularized with symmetric fourth-order dissipation, cf. \eqref{BK-System-BAK3-reg2}. The figures show that fourth-order dissipation with $0<\nu\ll1$ prevents the ill-posed growth observed in the unregularized system.}
    \label{fig: 1dbkbk_simulation}
\end{figure}

% \todo[inline]{DH \& Maneesh: Looks good! What is the argument of $\exp$? }

\begin{figure}[H]
    \centering
    \begin{subfigure}[t]{0.45\linewidth}
        \centering
        \includegraphics[width=\linewidth]{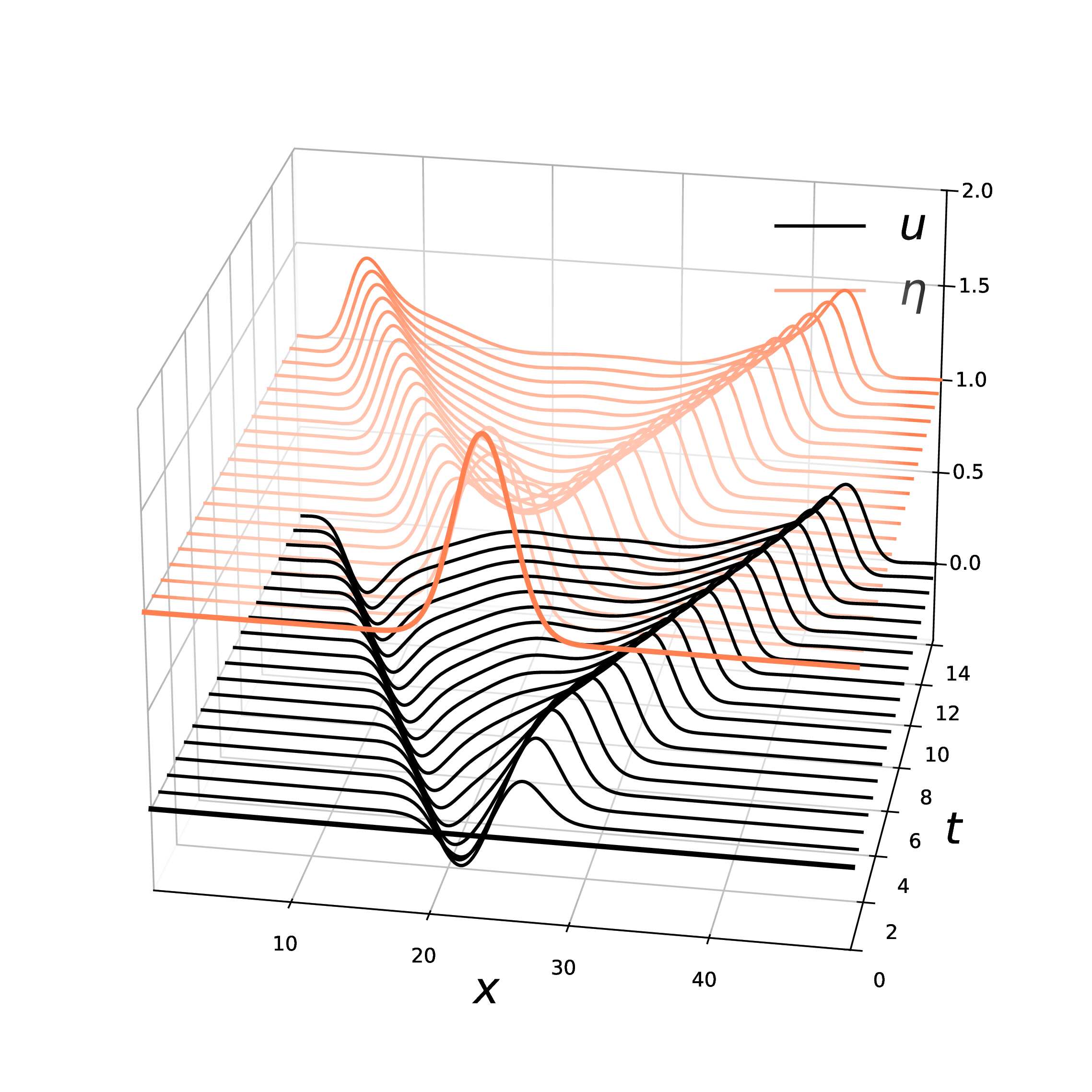}
        \caption{$\kappa =-0.5$}
        \label{fig: 1dbkbk_kappa-0.5}
    \end{subfigure}
    \hfill
    \begin{subfigure}[t]{0.45\linewidth}
        \centering
        \includegraphics[width=\linewidth]{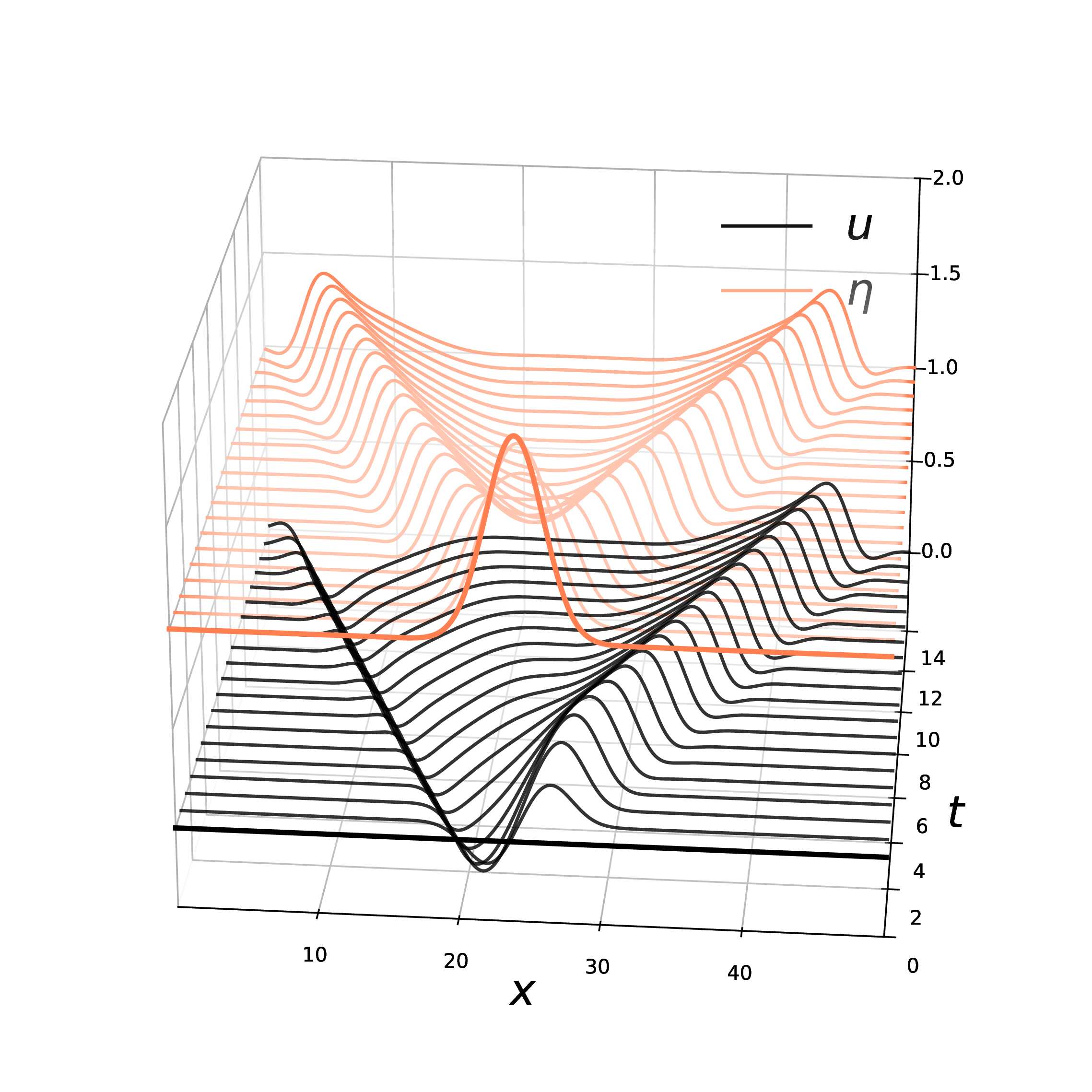}
        \caption{$\kappa=-0.1$}
        \label{fig: 1dbkbk_kappa-0.1}
    \end{subfigure}
    \caption{Waterfall plot of the regularized 1D BKBK system with symmetric fourth-order diffusion, cf. \eqref{BK-System-BAK3-reg2}. For the Gaussian initial condition $u=0$, $\eta=1+\exp(-(x-24)^2/8),\nu=0.01$, the solution rapidly splits into two dominant crests which propagate in opposite directions. While this bidirectional splitting occurs for both $\kappa= -0.5$ and $\kappa=-0.1$, the case $\kappa=-0.1$ in the right panel exhibits a small leading depression ahead of the main crest, with $\eta<\eta_0$, which is absent for $\kappa=-0.5$ in the left panel.}
    \label{fig: 1dbkbk_diffkappa}
\end{figure}

We use the nonsingular travelling wave \eqref{eq: bkbk_tw} as a benchmark initial condition (figure \ref{fig: BKBK_tw}), as reported in \cite{XIE200176}. 
\begin{equation}
\begin{split}
u(x, t)&=c-\lambda |\kappa| \tanh \left(\frac{\lambda}{2}(x-c t+\phi)\right)\,, \\
\eta(x, t)&=\frac{\lambda^2}{2} |\kappa|(|\kappa|+\kappa) \operatorname{sech}^2\left(\frac{\lambda}{2}(x-c t+\phi)\right)
\,,\end{split}
\label{eq: bkbk_tw}
\end{equation}
where $\lambda, c, \phi$ are constants. Although it is an exact solution of the unregularized BKBK system, direct time stepping of \eqref{BK-System-BAK3} is numerically unstable. The symmetric regularisation \eqref{BK-System-BAK3-reg2} provides a controlled setting for assessing (i) preservation of the wave profile at low wavenumbers, (ii) suppression of high-$k$ growth, and (iii) the dependence of phase speed and spectral energy transfer on dissipation.

\subsection{2D BKBK system simulation}
In the computational simulations of the two-dimensional BKBK system \eqref{2D BKBK}, we find that no fourth-order dissipation is required. Instead, we use a modified BKBK system derived from a regularised Hamiltonian constructed by augmenting the BKBK Hamiltonian \eqref{DSW-Ham-u} to include an energy penalty for large wave slope as presented in \cite{holm2025geometric}. Namely,
\begin{equation}
h(m, \eta) = \int_M \frac{|\mathbf{m}|^2}{2 \eta}-\kappa \mathbf{m} \cdot \nabla \ln (\eta)+\frac{1}{2}\left(\eta^2+\alpha^2|\nabla \eta|^2\right) d^2 x    
\,,\label{DSW-Ham-reg}
\end{equation}
whose variational derivatives are given by
\begin{align}
\begin{split}
    \frac{\delta h}{\delta \mathbf{m}} 
    & =\mathbf{m} / \eta-\kappa \nabla \ln (\eta)
    =\mathbf{u}-\kappa \nabla \ln (\eta) =: \bv, \\
\frac{\delta h}{\delta \eta} & =-\frac{|\mathbf{m}|^2}{2 \eta^2}+\frac{\kappa}{\eta} \operatorname{div} \mathbf{m}+\mathrm{g}\left(1-\alpha^2 \Delta\right) \eta =: -{\cal B} 
\,.
\end{split}
\label{Ham-DSW}
\end{align}
The regularised Hamiltonian \eqref{DSW-Ham-reg} mimics the Hamiltonian formulation of the Lagrangian Averaged Navier Stokes $\alpha$ (LANS-$\alpha$) subgrid scale model \cite{CFHOTW1998} in which the gradient energy penalty is on the velocity $u$, rather than the height $\eta$.
This energy penalty introduces nonlinear dispersion, which suppresses high-wavenumber activity without artificial viscosity, thereby preserving the conservative/geometric structure (energy, Kelvin circulation, and potential-vorticity conservation in the classic shallow water model). As shown below in the results of numerical simulations, coherent vortical features remain sharp and the high-wavenumber blow-up seen in 1D is controlled.
% Thus, in 2D the BKBK system may be regularised effectively by an energy penalty for developing high wave slope rather than by fourth-order diffusion as for BKBK in 1D.

After transforming from the Lie--Poisson operator form \eqref{BKBK-Poisson2}
the regularised 2D BKBK system may be expressed in fluid dynamics notation as
\begin{align}
\begin{split}
\partial_t \mathbf{u}&+(\mathbf{v} \cdot \nabla) \mathbf{u}+u_j \nabla v^j  =\nabla {\cal B}, \\
\partial_t \eta &+\operatorname{div}(\eta \mathbf{v})  =0, \qquad 
\bv := \mathbf{u}-\kappa \nabla \ln (\eta) 
\,.\end{split}
\label{2DBKBK-system}
\end{align}
The 2D BKBK system \eqref{2DBKBK-system} represents shallow water dynamics with modified transport velocity $\mathbf{v}$ and Bernoulli function $B(\mathbf{u}, \eta)$ defined in \eqref{Ham-DSW}. The remainder of this section provides examples of their solution behaviour. Section \ref{sec-4} reprises their variational derivation. 

We solve the two-dimensional shallow water equations using the pseudo-spectral method in a periodic domain of size $L_x \times L_y=16 \times 16$ with resolution $192 \times 192$. Time integration uses the semi-implicit SBDF2 scheme with a timestep $10^{-6}$, up to a final time $t=2.0$. Model parameters are $g=1, \alpha=0.02$, $\kappa=-0.5$, and $\kappa=-0.05$. The background depth is $\eta_0=4.0$ with no mean flow. The initial free-surface height is given by two positive Gaussian ridges of amplitude $h_0=4.0$ (width $\sigma=0.7$ ) centered at $(x, y)=\left(L_x / 2 \pm \delta x, 0\right)$, with $\delta x=1.1 \sigma$, together with two weaker negative Gaussian anomalies of amplitude $-0.01 h_0$ at $(x, y)=\left(L_x / 2, \pm \delta y\right)$, with $\delta y=1.7 \sigma$. The velocity field is initialized in geostrophic balance $u = -g \hat{\mathbf{z}} \times \nabla \eta/f_0$ with $f_0=50$.

\begin{figure}[H]
    \centering
        \includegraphics[width=1\linewidth]{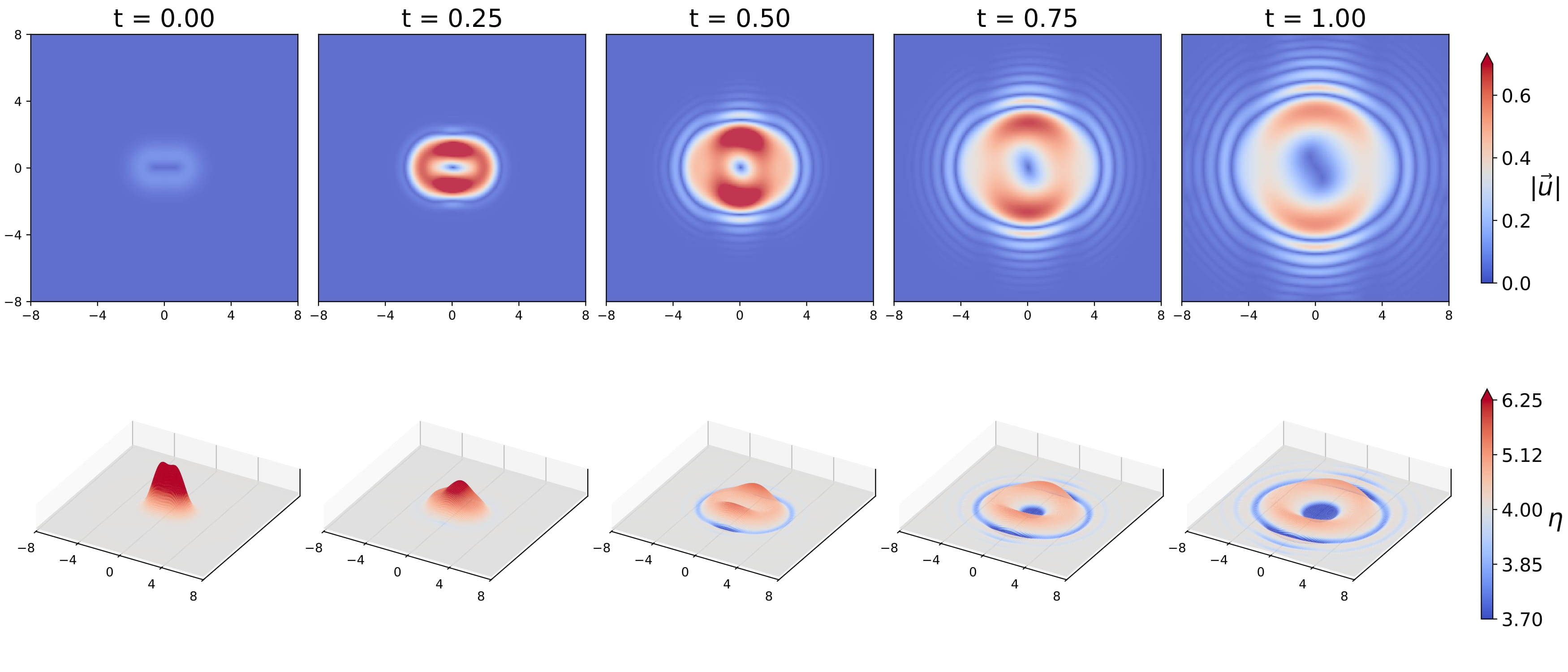}
        \caption{Simulation of the 2D BKBK equation initialised with two Gaussian ridges, for negative $\kappa=-0.05$. Top: velocity magnitude $|\vec{u}|$. Bottom: surface height $\eta$ as a 3D surface. At $t=0$, the free surface consists of two initial Gaussian ridges. By $t=0.25$, the two peaks merge into a single crest, generating outward-propagating ring waves. At $t=0.50$, the surface splits along the $y$-direction, producing two secondary peaks while the ring pattern continues to expand. At $t=0.75$, a depression forms at the center, accompanied by strong outward-propagating velocity oscillations; the free-surface height also exhibits significant fluctuations.}
    \label{fig: gaussian_005}
\end{figure}

\begin{figure}[H]
    \centering
        \includegraphics[width=1\linewidth]{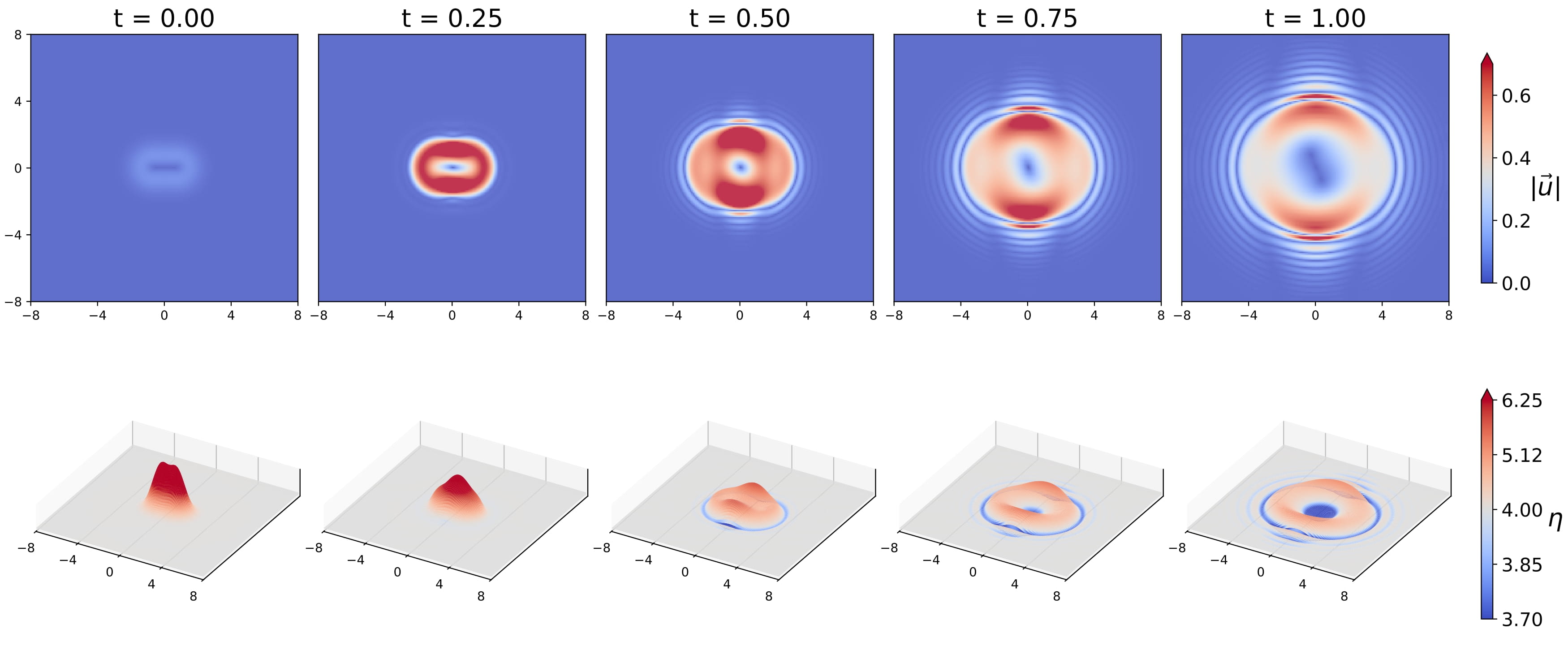}
        \caption{Comparison of simulations with different values of $\kappa=-0.5$. The outer rings are smaller at $t=1$, indicating slower outward wave propagation compared to $\kappa= -0.05$. At $t=0.25$, the surface profile also collapses more slowly, and the resulting ring structures remain more compact.}
    \label{fig: gaussian_05}
\end{figure}

We solve the two-dimensional shallow water equations in a periodic domain of size $L_x \times L_y=24 \times 16$, discretized with $384 \times 256$ grid points. The initial condition is a localized perturbation of velocity in the $x$ component, given by $u_x(x, y) = W(-8/3, 8/3 ; 0.5 ; y) W(10,11 ; 0.5 ; x)$ where $W(a, b ; \delta ; z)=\frac{1}{2}\left[\tanh \left(\frac{z-a}{\delta}\right)-\tanh \left(\frac{z-b}{\delta}\right)\right]$ is a smooth rectangular window. Other parameters and numerical settings are the same as in the previous experiments.

\begin{figure}[H]
    \centering
    \includegraphics[width=1\linewidth]{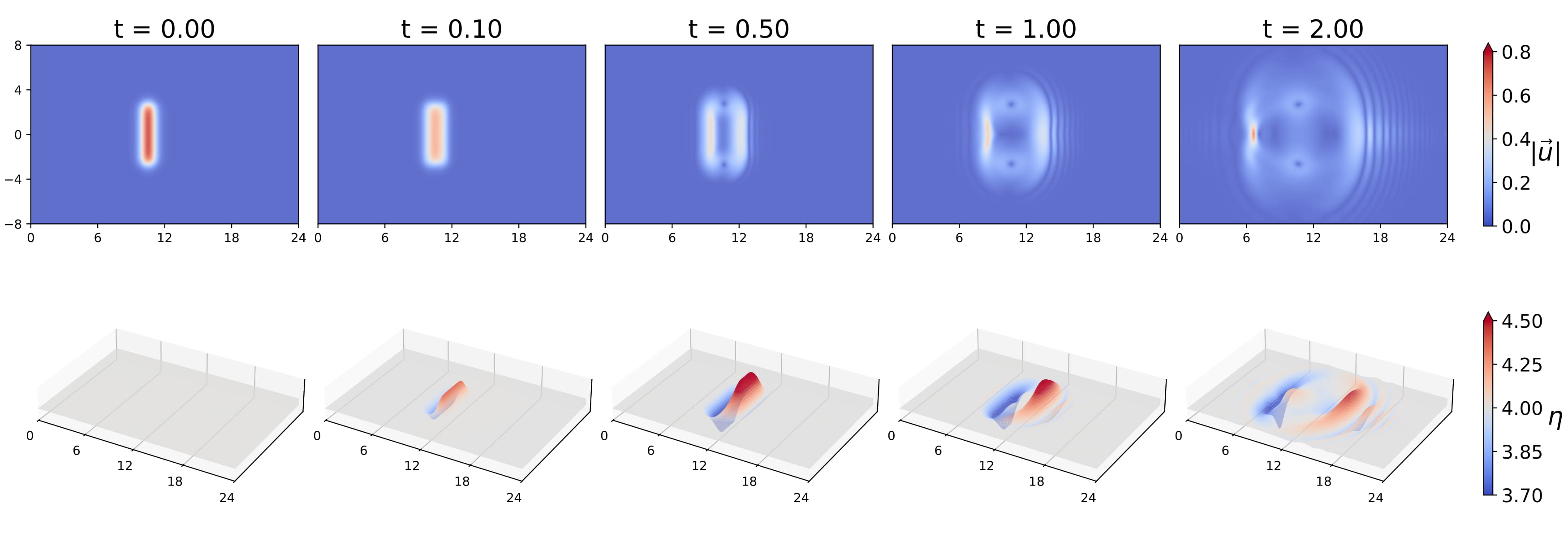}
    \caption{Simulation of the 2D BKBK equation with a localized tanh-segment initialization, for negative $\kappa=-0.5$. Top: velocity magnitude $|\vec{u}|$. Bottom: surface height $\eta$ as a 3D surface. At $t= 0$, the flow is initialized with a narrow rectangular segment in $u_x$. By $t=0.1$, the disturbance begins to tilt and induces small surface deflections. At $t=0.5$, the velocity field splits into two lobes with a vortex dipole while the surface height develops two crests and moves in the opposite direction. At $t= 1$, outward-propagating oscillations appear and the central structure weakens. By $t= 2$, the disturbance has radiated into compact concentric wavefronts.}
    \label{fig: seg_neg}
\end{figure}
\begin{figure}[H]
    \centering
    \includegraphics[width=1\linewidth]{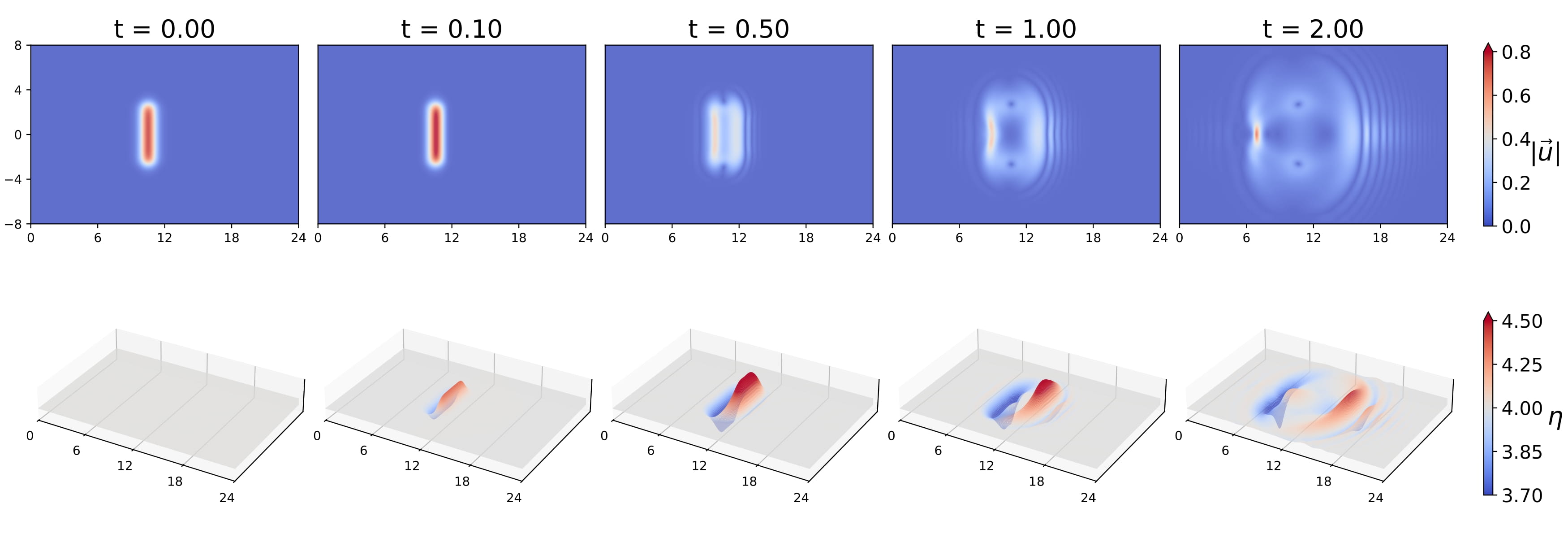}
    \caption{Comparison of simulations on tanh-segment with $\kappa=+0.5$. At $t=0.10$, the velocity segment does not diffuse but instead intensifies, as indicated by the stronger red shading. The subsequent evolution resembles the $\kappa=-0.5$ case, but the overall development proceeds more slowly.}
    \label{fig: seg_pos}
\end{figure}

% \begin{figure}[H]
%     \centering
%     \includegraphics[width=0.9\linewidth]{fig/phase_eta_shifted_xy}
%     \caption{In this figure, the upper (resp, lower) panels represent wave elevation (resp, velocity). Initially, two segments of radial tanh-shaped ridges in wave-elevation of different lengths meet at an angle in still water with $\eta=4$. For $\kappa=0.5$, these two segments each generate waves which propagate by expanding in opposite directions away from the initial orientation. The diffractive mutual interactions occuring during their propagation generate complex wave structures in amplitude and velocity which take similar forms for both positive and negative values of $\kappa$. Thus, a local change in curvature in 2D BKBK wave elevation can be expected to generate complex wave activity. }
%     \label{fig:5}
% \end{figure}

\begin{figure}[H]
    \centering
    \includegraphics[width=1\linewidth]{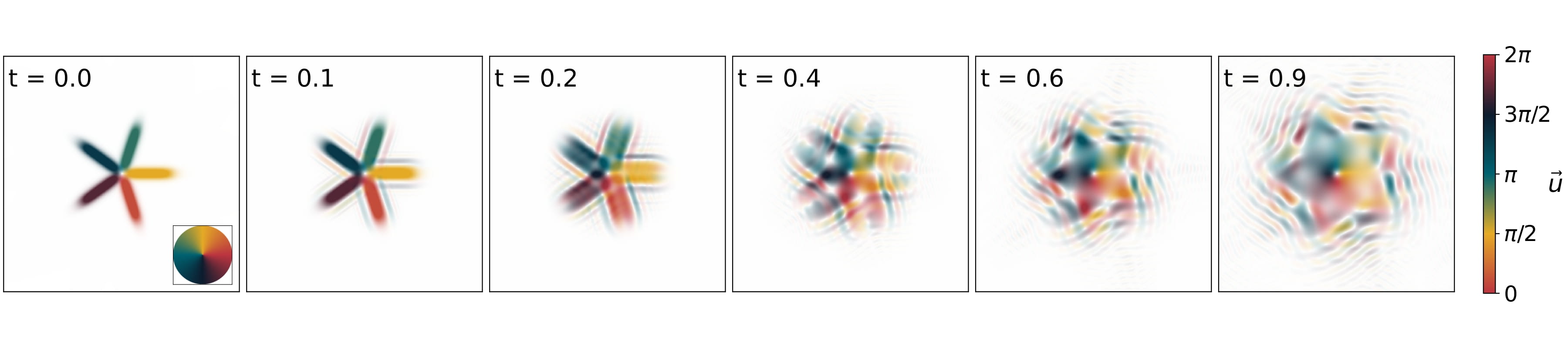}
    \caption{This figure shows snapshots of a 2D BKBK simulation for five-fold symmetric radial tanh-segments with counter-clockwise tangential flow and $\eta=4$ and $\kappa=0.5$. Colours encode the orientation $\theta \in[0,2 \pi)$ of the 2D velocity $\bu$; saturation/opacity indicates the speed $|\bu|$. Snapshots at $t=0,0.1,0.2,0.4,0.6,0.9$ are shown. At $t=0$, the field is initialized with  five radial tanh-segments with counter-clockwise tangential flow and $\eta=4$. Over time, the pattern disperses and develops ring-like/rippled fine scales, reflecting increased directional shear and multiscale interactions.}
    \label{fig: star_kappa0.5}
\end{figure}
\begin{figure}[H]
    \centering
    \includegraphics[width=1\linewidth]{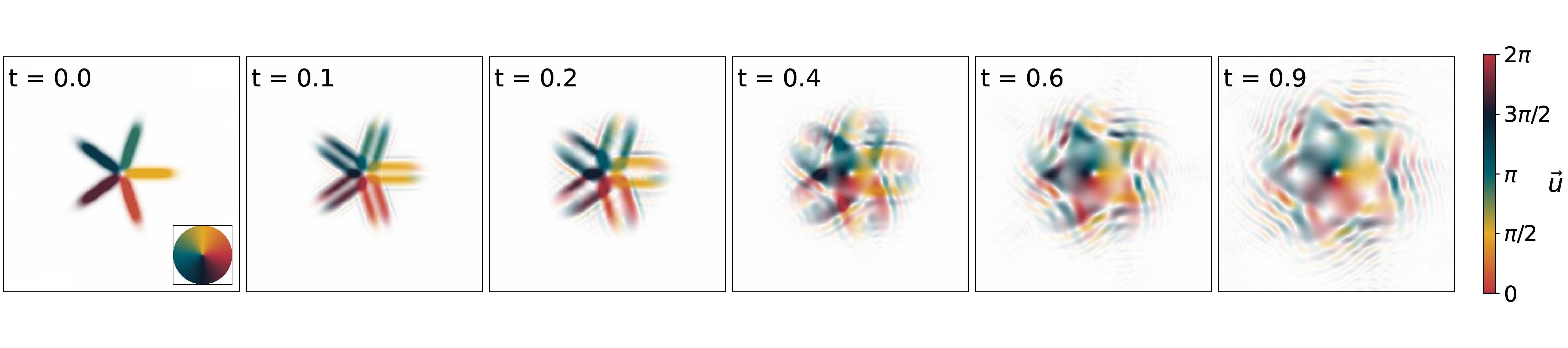}
    \caption{This figure shows a 2D BKBK simulation of five-fold symmetric radial tanh-segments of wave-elevation with $\kappa=-0.5$, in the same initial configuration as in Fig. \ref{fig: star_kappa0.5}. In comparison, negative $\kappa$ exhibits faster radial spreading. However, the differences in solution behaviour depending on the $\pm$ sign of $\kappa$ at early times tend to subside at later times. The tendency for $\pm\kappa$ similarity in solution behaviour at later times is consistent with the $\kappa^2$ dependence of the linearised Lyapunov stability results derived in section \ref{sec-5}. Remarkably, the five-fold symmetry is preserved and local changes in curvature of wave phase arise at the outer edges which may be expected to generate further changes in shape. \\
    }
    \label{fig: star_kappa-0.5}
\end{figure}
% \todo[inline]{DH: Hi Hanchun, would you like to chat for 5 minutes to catch up a bit?\\
% HW: Hi Darryl, of course. You can call me on teams. Hanchun. \\
% Your figures look good, Hanchun. The tendency for $\pm\kappa$ similarity in solution behaviour at later times is consistent with the $\kappa^2$ dependence of the linearised Lyapunov stability results in section \ref{sec-5}.  }
%%%%%%%%%%%%%%%%%%%%%%%%%%%%%%%%%%%%%%%%%%%%%%%%%%%%%%%%

\section{Summary conclusion and open questions}\label{sec-6}
%%%%%%%%%%%%%%%%%%%%%%%%%%%%%%%%%%%%%%%%%%%%%%%%%%%%%%%%
\begin{comment}
As stated in Remark \ref{rem:NLS}, the representation of the BKBK equation in terms of transport velocity
$v$ is equivalent under the inverse Madelung transformation to the focusing NLS equation 
for an imaginary parameter $\kappa = i/2$. Hence, the analytical and computational investigations here of the 1D and 2D BKBK system may be regarded as being complementary to analogous investigations of modulation instability via the 1D and 2D NLS equations. 
%
 The 2D BKBK system \eqref{2DBKBK-system} derived in section \ref{sec-4} modifies the transport velocity $\mathbf{v}$ and Bernoulli function $B(\mathbf{u}, \eta)$ of classical shallow water dynamics, as defined in \eqref{Ham-DSW}. After introducing these equations, section \ref{sec-3} demonstrates their solution behaviour in several examples. Section \ref{sec-4} reprises their variational derivation. Section \ref{sec-5} derives the conditions for equilibrium solutions of the 
 2D BKBK system \eqref{2DBKBK-system} to be linearly Lyapunov stable.
\end{comment}
%%%%%%%%%%%%%%%%%%%%%%%%%%%%%%%%%%%%%%%%%%%%%%%%%%%%%%%%

One of the first observations in this paper is that the BKBK system introduces a dynamical shift in the transport velocity in the motion equations and the advection equations in both the 1D and 2D BKBK systems. This dynamical shift in the transport velocity is proportional to the real constant $\kappa$, and when $\kappa=0$ one formally recovers the classical shallow water wave equations. See, e.g., equation \eqref{BK-System-BAK3}. 

However, the BKBK shift in transport velocity comprises a singular perturbation which also introduces backwards diffusion in the motion equation. As shown in equation \eqref{disp-noalpha} for 1D BKBK, this feature produces high-wavenumber instabilities whose exponential growth rates increase as $\omega(k) = \pm ik\sqrt{1-\kappa^2k^2}$. Thus, although the 1D BKBK system is known to be completely integrable \cite{kaup1975higher, kupershmidt1985mathematics}, its solution behaviour is ill-posed for real values of $\kappa$. One notices, though, that this high-wavenumber instability would not occur if the parameter $\kappa$ were imaginary, instead of being real.

Section \ref{sec-2} reframes the 1D BKBK system as a Lie--Poisson Hamiltonian system 
and then explores its relation with modulation instability \cite{peregrine1983water}, by noticing that %the BKBK Hamiltonian $\mathfrak{h}(v,\eta)$ in \eqref{Ham-v} is similar to the Hamiltonian of the Nonlinear Schr\"odinger (NLS) equation, when expressed in 
the Madelung transformation  
%
%$(\eta, \phi)$ related to the complex wave function $\psi$ by $\psi = \sqrt{\eta}\exp(i\phi)$ and $v = \p_x \phi$ \cite{M1927}. 
%In fact, after choosing $\kappa = {i}/{2}$, the change of variables $\eta v=\eta \p_x \phi ={\rm Im}(\psi^*\p_x \psi)$ and $\eta=\|\psi\|^2$ 
%
transforms the BKBK Hamiltonian in \eqref{Ham-v} into the Hamiltonian for the focusing NLS equation. %with complex canonical variables $(\psi,\psi^*)$ and Planck constant $\hbar$ set equal to unity. 
%Inserting $\kappa = {i}/{2}$ into the dispersion relation $\omega^2(k^2)$ in \eqref{disp-noalpha} also changes the BKBK instability for real $\kappa$ into a travelling wave with dispersion relation, 
%\begin{equation}
    %\omega^2(k^2) = k^2\left(1 + \tfrac14 k^2\right)  \,. 
%\end{equation}
%
Given this observation, studies of instability of water waves using the 1D BKBK system with real $\kappa$ may be regarded as being complementary to the traditional study of modulation instability of water waves using the 1D NLS equation as in \cite{peregrine1983water}. 
%The Madelung transform of the 1D BKBK system for $\kappa=i/2$ to the focusing NLS equation also applies for the 2D BKBK system, but only in the case of potential flow, as in equation \eqref{2D-potentflow}.
%The BKBK system has been associated with the NLS equation in 1D and in 2D for potential flow with imaginary $\kappa$. Nonetheless, 
%To make progress toward computationally simulating the solution behaviour of BKBK with real $\kappa$, some regularisation of BKBK must be considered. 

In preparation for developing ideas for regularisation of the BKBK system to help meet the challenges of its computational simulations, section \ref{sec-3} explains the variational derivations of the 2D BKBK system, as well as its Hamiltonian structures and their implications such as the Lagrangian advection of potential vorticity and its associated conservation laws.

Section \ref{sec-4} uses the Hamiltonian structures investigated in section \ref{sec-3} to develop a classification of the equilibrium solutions of the 2D BKBK system and to determine the conditions for the linear Lyapunov stability of these equilibria by using the energy-Casimir approach \cite{holm1985nonlinear}. The Lyapunov stability conditions obtained from the energy-Casimir approach extend the class of 2D BKBK equilibria to finite velocity and match the instability results obtained by linearising the BKBK equations around equilibria with vanishing velocity and constant elevation.  

In section \ref{sec-5}, the figures for the computational simulations of the regularised 1D BKBK system \eqref{BK-System-BAK3-reg2} show that introducing fourth-order dissipation with coefficient $0<\nu\ll1$ prevents the ill-posed growth observed in the 1D BKBK unregularised system with $\nu=0$. 
This regularisation of the 1D BKBK system by introducing fourth-order dissipation has the advantage that it may also introduce a new variant of the Kuramoto--Shivashinsky equation \cite{cross1993pattern} which may have interesting low dimensional solution behaviour for 1D BKBK shallow water wave dynamics, although this investigation would be for future research. 

In 2D, the BKBK transport velocity shift also appears in the material loop velocity in the 2D Kelvin circulation theorem as well as in the transport of the potential vorticity, as discussed in section \ref{sec-4}. In studying 2D computational simulations of the BKBK system \eqref{2DBKBK-system} in section \ref{sec-5}, we used a regularised Hamiltonian accomplished by augmenting the 2D BKBK Hamiltonian to include an energy penalty for large wave slope. See equation \eqref{DSW-Ham-reg}. 

Section \ref{sec-5} shows the first simulations of the regularised BKBK dynamics derived here in both 1D and 2D. These simulations represent proof of principle and indication of interesting dynamics of the regularised BKBK flows. Further investigations of the sensitivity of the solutions to the magnitude and sign of $\kappa$, for example, are deferred for future research. 

\color{black}%%%%%%%%%

\subsection*{Acknowledgements.}
DH is grateful for the time he enjoyed with his friends David Kaup and Boris Kupershmidt collaborating together on the topic of integrable Hamiltonian systems.  During this work DH and HW were partially supported by European Research Council (ERC) Synergy grant Stochastic Transport in Upper Ocean Dynamics (STUOD) -- DLV-856408. The work of RH is partially supported by the Office of Naval Research (ONR) grant award N00014-22-1-2082, Stochastic Parameterization of Ocean Turbulence for Observational Networks. 

\bibliographystyle{alpha}
\bibliography{main.bib}

\end{document}